\documentclass[12pt,leqno,fleqn]{amsart}
\usepackage{latexsym}
\usepackage{amssymb}
\usepackage{amsxtra}
\usepackage{amsmath}
\usepackage{amsfonts}
\usepackage{CJK}
\newcommand{\cleqn}{\setcounter{equation}{0}}
\newcommand{\clth}{\setcounter{theorem}{0}}
\newcommand {\sectionnew}[1]{\section{#1}\cleqn\clth}
\newtheorem{theorem}{Theorem}[section]
\newtheorem{lemma}[theorem]{Lemma}
\newtheorem{proposition}[theorem]{Proposition}
\newtheorem{corollary}[theorem]{Corollary}
\newtheorem{definition}[theorem]{Definition}
\renewcommand{\P}{\mathcal{P}}

\newcommand{\M}{\mathcal{M}}
\renewcommand{\L}{\mathcal{L}}
\newcommand{\Z}{\mathbb{Z}}

\newcommand{\A}{\mathcal{A}}
\newcommand{\B}{\mathcal{B}}

\def\res{\mathop{\rm Res}\nolimits}
\def\({\left(}
\def\){\right)}
\def\[{\left[}
\def\]{\right]}
\def\d{\partial}

\def\ep{\epsilon}

\def\La{\Lambda}
\def\d{\partial}
\makeatletter
    
    \newcommand{\Rmnum}[1]{\expandafter\@slowromancap\romannumeral #1@}
  \makeatother

\begin{document}

\title{Dispersionless bigraded Toda Hierarchy and its additional symmetry}
\author{Chuanzhong Li\dag\S
, Jingsong He\dag$^*$ } \dedicatory {  \dag Department of
Mathematics,  NBU, Ningbo, 315211, Zhejiang, P.\ R.\ China\\
 \S Department of Mathematics, USTC, Hefei, 230026 Anhui, P.\ R.\ China}

\thanks{$^*$ Corresponding author: hejingsong@nbu.edu.cn}
\texttt{}

\date{}

%%%%%%%%%%%%%%%%%%%%%%%%%%%%%%%%%%%%%%%%%%%%%%%%
\begin{abstract}
 In this paper, we firstly give the definition of dipersionless
bigraded Toda hierarchy (dBTH) and introduce some Sato theory on
dBTH. Then we define Orlov-Schulman's $\M_L$, $\M_R$ operator and
give the additional Block symmetry  of  dBTH. Meanwhile we give
tau function of dBTH and some some related dipersionless bilinear
equations.

\end{abstract}

%%%%%%%%%%%%%%%%%%%%%%%%%%%%%%%%%%%%%%%%%%%%%%%%

\maketitle
\tableofcontents
\allowdisplaybreaks
 \setcounter{section}{0}

\sectionnew{Introduction}
The Toda lattice equation is a nonlinear evolutionary
differential-difference equation introduced by Toda \cite{Toda}
describing an infinite system of masses on a line that interact
through an exponential force. This
equation is completely integrable, i.e. admits infinite conserved
quantities and  has  important applications in many different fields
such as classical and quantum field theory, in particular in the
theory of Gromov-Witten invariants \cite{Z}.
 Considering its application to 2D topological field theory \cite{D witten,witten}, one  extended the interpolated Toda lattice
hierarchy into the so-called extended Toda hierarchy \cite{CDZ,DZ
virasora}. In the paper \cite{C}, it generalized the Toda lattice
hierarchy(TH) and extended the Toda lattice hierarchy by considering
$N+M$ dependent variables and used them to provide a Lax pair
definition of the extended bigraded Toda hierarchy(EBTH). This
hierarchy later lead to a series of results
\cite{TH,ourJMP,ourBlock,solutionBTH}. In fact that model has been proposed
in \cite{CMP kodama} because  the dispersionless EBTH can be
obtained from the dispersionless KP hierarchy and the dispersionless
Toda (dToda) hierarchy  describes the genus zero-limit of the
Landau-Ginzburg formulation of two-dimensional string theory
\cite{Dijkgraaf,Hanany,Takasakicmp} .

Additional symmetries have been analyzed in the explicit form of the
additional flows of KP hierarchy  given by Orlov and Shulman
\cite{os1}. This kind of additional  flows include dynamic variables
explicitly. That additional symmetries  form a centerless
$W_{1+\infty}$ algebra which is closely related to matrix
model\cite{D witten,Douglas} because of the Virasoro constraint and
string equation. About Toda hierarchy, there was parallel results
\cite{Takasakicmp} which was used to give string equations
 and Riemann-Hilbert problem of dispersionless Toda hierarchy \cite{Takasaki}.
 Because of the close reduction relation between Toda hierarchy and the bigraded Toda
 hierarchy(BTH), it motivated us to consider the additional symmetry of the
 BTH. In another paper\cite{ourBlock}, we give a novel Block type
additional symmetry of the BTH. This is the first time to find the
direct relation between  integrable hierarchy and the Block type
algebra. The representation theory of the Block type infinite
algebra has been studied intensively in
references\cite{Block}-\cite{Su}.

Dispersionless integrable systems have been found very important in
the study of all kinds of nonlinear phenomenon in various fields of
physics and mathematics, particularly in the application in
topological field theory \cite{Takasakicmp} and matrix model
theory\cite{Kodama-Pierce,matrix model}. In particular, the
dispersionless integrable systems have many typical properties as
usual integrable systems such as Lax pair, infinite conservation
law, symmetry and the Hirota bilinear
equations(HBEs). There is a dispersionless limit to get the dBTH from the
BTH. However, after taking this limit of the BTH, whether  the Block type Lie algebraic structure can be preserved is still
an interesting question.  So the high relevance of the  TH and BTH
in mathematical physics motivates us to focus on the HBEs and the additional symmetry of the dBTH in this
paper.

The paper is organized as follows.  In Section 2   the
definition of dBTH and corresponding dispersionless version of Sato theory are introduced. In Section 3, we define Orlov-Schulman's $\M_L$,
$\M_R$ function. The Block type  additional symmetry of dBTH will be given
in Section 4. In Section 5, we give the quasi-classical limit  of
BTH to get  dBTH. In Section 6, we give the dispersionless Hirota
bilinear identity of  dBTH which provide a very sound mathematical background in its possible applications.
Section 7 is devoted to conclusions and
discussions.

\sectionnew{ The dBTH }
 Introduce
firstly the lax operator of dispersionless bigraded Toda hierarchy (dBTH) as following
\begin{equation}\label{Lax operator}
\L=k^{N}+u_{N-1}k^{N-1}+\dots + u_{-M}
k^{-M}
\end{equation}
( $N,M \geq1$ are two fixed positive integers).
  The variables $u_j$ are functions of the real variable $x$. The Lax operator $\L$ can be
written in two different dressed ways
  $$\L=e^{ad\varphi_L}(k^N )= e^{ad\phi}e^{ad\varphi_R}(k^{-M}),$$
  which in fact gives  the constraint(string equation)\cite{string equations} of the two dimensional dispersionless Toda hierarchy.
The two dressing functions has the following form \begin{eqnarray}
&& \varphi_L=w_1k^{-1}+w_2k^{-2}+\ldots,
\label{dressP}\\
&& \varphi_R=\tilde{w_1}k+\tilde{w_2}k^2+ \ldots, \label{dressQ}
\end{eqnarray}
We can get the following relation
 \begin{eqnarray}
&& u_{N-1}=-N\frac{\d w_1}{\d x}, \ \ u_{-M}(x)=\phi(x+M).
\end{eqnarray}
  The
pair is unique up to adding some laurent serials about variable $k$ with  coefficients which do not depend on $x$.

 The dispersionless  bigraded Toda hierarchy can be defined as following.
\begin{definition} \label{deflax}
The dispersionless  bigraded Toda hierarchy(dBTH) consists of the system
of flows given in the Lax pair formalism by
\begin{equation}
  \label{edef2}
  \frac{\d \L}{\partial t_{\gamma, n}} = \{ \A_{\gamma,n} ,\L \}:=k(\frac{\partial  \A_{\gamma,n}}{\partial k}\frac{\partial \L}{\partial x}-\frac{\partial  \A_{\gamma,n}}{\partial x}\frac{\partial \L}{\partial k}),
\end{equation}
for $\gamma = N,N-1,N-2, \dots , -M+1$ and $n \geq 0$. The operators
$A_{\gamma ,n}$ are defined by
\begin{subequations}
\label{Adef}
\begin{align}
  &\A_{\gamma,n} =  ( \L^{n+1-\frac{\gamma-1}N })_+ \quad \text{for} \quad \gamma = N,N-1, \dots, 1\\
  &\A_{\gamma,n} = - ( \L^{n+1+\frac{\gamma}M })_- \quad \text{for} \quad \gamma = 0, \dots,
  -M+1.
\end{align}
\end{subequations}
\end{definition}

 Particularly for $N=1=M$ this hierarchy
coincides with the dispersionless  Toda hierarchy.

To see the dBTH clearly, we will introduce two examples as following, i.e. (1,2)-dBTH and (2,2)-dBTH.
\subsection{Example as the (1,2)-dBTH}
  The Lax operator is
\begin{equation}\L=k+u_0 + u_{-1}
k^{-1}+ u_{-2}k^{-2}.
\end{equation}
Then there  will be one fraction power of $\L$, denoted as $\L^{\frac12}$  as following form

\begin{equation}\L^{\frac12}=b_{-1}k^{-1}+b_0 + b_{1}
k+ b_{2}k^{2}+\dots.
\end{equation}
We can get some relations of $\{ b_j; j\geq -1\}$ with $\{u_i; -M\leq i\leq N-1\}$ as following
\begin{eqnarray}\label{a operator12} b_{-1}^{2}= u_{-2}.
 \end{eqnarray}
Then by Lax equation, we get the $t_{1,0}$ flow of (1,2)-BTH which is equivalent to $t_{0,0}$ flow as following
\begin{eqnarray}
 \partial_{1,0} \L= \{k+u_0, \L\}
\end{eqnarray}
which corresponds to
\begin{eqnarray}\label{1,0 flow12}
\begin{cases}
 \partial_{1,0} u_0&=  \frac{\d u_{-1} }{\d x}\\
 \partial_{1,0} u_{-1}&=  \frac{\d u_{-2}}{\d x}+u_{-1}\frac{\d u_{0}}{\d x}\\
  \partial_{1,0} u_{-2}&=2\frac{\d u_{0}}{\d x}u_{-2}.
 \end{cases}
 \end{eqnarray}

By Lax equation, we get the $t_{-1,0}$ flow of (1,2)-BTH
\begin{eqnarray}
 \partial_{-1,0} \L= \{-u_{-2}^{\frac12}k^{-1}, \L\}
\end{eqnarray}
which correspond to
\begin{eqnarray}\label{-1,0 flow12}
\begin{cases}
 \partial_{-1,0} u_0&=  \frac{\d u_{-2}^{\frac12} }{\d x}\\
 \partial_{-1,0} u_{-1}&=  u_{-2}^{\frac12}\frac{\d u_{0}}{\d x}\\
  \partial_{-1,0} u_{-2}&=\frac{\d u_{-1}}{\d x}u_{-2}^{\frac12}-\frac{\d u_{-2}^{\frac12}}{\d x}u_{-1}\\
0&=\frac{\d u_{-2}}{\d x}u_{-2}^{\frac12}-2\frac{\d u_{-2}^{\frac12}}{\d x}u_{-2}.
 \end{cases}
 \end{eqnarray}
\subsection{Example as the (2,2)-dBTH}
  The Lax operator is
\begin{equation}\L=k^2+u_{1}k+u_0 + u_{-1}
k^{-1}+ u_{-2}k^{-2}.
\end{equation}
Then there  will be two different fraction power of $\L$, denoted as $\L_N^{\frac12}$ and $\L_M^{\frac12}$ respectively as following form
\begin{equation}\L_N^{\frac12}=k+a_0 + a_{-1}
k^{-1}+ a_{-2}k^{-2}+\dots,
\end{equation}
\begin{equation}\L_M^{\frac12}=a'_{-1}k^{-1}+a'_0 + a'_{1}
k+ a'_{2}k^{2}+\dots.
\end{equation}
We can get some relations of $\{a_i; i\leq 0\},\{ a'_j; j\geq -1\}$ with $\{u_i; -M\leq i\leq N-1\}$ as following
\begin{eqnarray}\label{a operator}
 u_{1}=2a_0,\ \ \ a_{-1}^{'2}= u_{-2}.
 \end{eqnarray}
Then by Lax equation, we get the $t_{2,0}$ flow of (2,2)-BTH
\begin{eqnarray}
 \partial_{2,0} \L= \{k +\frac12u_1, \L\}
\end{eqnarray}
which corresponds to
\begin{eqnarray}\label{2,0 flow}
\begin{cases}
 \partial_{2,0} u_1&= \frac{\d u_{0}}{\d x}-\frac12\frac{\d u_1}{\d x}\\
 \partial_{2,0} u_0&=  \frac{\d u_{-1}}{\d x}\\
 \partial_{2,0} u_{-1}&= \frac{\d u_{-2}}{\d x}+\frac12\frac{\d u_{-1}}{\d x}\\
  \partial_{2,0} u_{-2}&=\frac{\d u_{-2}}{\d x}.
 \end{cases}
 \end{eqnarray}
Similarly by Lax equation, we get the $t_{1,0}$ flow of (2,2)-BTH which is equivalent to $t_{0,0}$ flow as following
\begin{eqnarray}
 \partial_{1,0} \L= \{k^{2}+u_{1}k+u_0, \L\}
\end{eqnarray}
which corresponds to
\begin{eqnarray}\label{1,0 flow}
\begin{cases}
 \partial_{1,0} u_1&= 2\frac{\d u_{-1}}{\d x}\\
 \partial_{1,0} u_0&=  2\frac{\d u_{-2} }{\d x}+2u_1\frac{\d u_{-1} }{\d x}-u_1\frac{\d u_{0} }{\d x}\\
 \partial_{1,0} u_{-1}&=  u_{1}\frac{\d u_{-2}}{\d x}+2u_{-2}\frac{\d u_{1}}{\d x}+u_{-1}\frac{\d u_{0}}{\d x}\\
  \partial_{1,0} u_{-2}&=2\frac{\d u_{0}}{\d x}u_{-2}.
 \end{cases}
 \end{eqnarray}

By Lax equation, we get the $t_{-1,0}$ flow of (2,2)-BTH
\begin{eqnarray}
 \partial_{-1,0} \L= \{-u_{-2}^{\frac12}k^{-1}, \L\}
\end{eqnarray}
which corresponds to
\begin{eqnarray}\label{-1,0 flow}
\begin{cases}
 \partial_{-1,0} u_1&= 2\frac{\d u_{-2}^{\frac12}}{\d x}\\
 \partial_{-1,0} u_0&=  \frac{\d u_{-2}^{\frac12} u_{1}}{\d x}\\
 \partial_{-1,0} u_{-1}&=  u_{-2}^{\frac12}\frac{\d u_{0}}{\d x}\\
  \partial_{-1,0} u_{-2}&=\frac{\d u_{-1}}{\d x}u_{-2}^{\frac12}-\frac{\d u_{-2}^{\frac12}}{\d x}u_{-1}\\
0&=\frac{\d u_{-2}}{\d x}u_{-2}^{\frac12}-2\frac{\d u_{-2}^{\frac12}}{\d x}u_{-2}.
 \end{cases}
 \end{eqnarray}

 For the convenience
to lead to the Sato equation, we will  define the following
functions:

\begin{equation}
  \B_{\gamma , n} :=
\begin{cases}
  \L^{n+1-\frac{\gamma-1}{N}} &\gamma=N\dots 1\\
 \L^{n+1+\frac{\gamma}{M}} &\gamma = 0\dots -M+1.
  \end{cases}
\end{equation}
\begin{proposition} \label{lemD}
The following identities hold true
\begin{align}
  \label{d6}
   & (\L^{\frac1N})_{t_{\gamma,p}} =   \{- (\B_{\gamma,p})_-, \L^{\frac1N}\} \\
  \label{d6i} &(\L^{\frac1M})_{t_{\gamma,p}} =  \{(\B_{\gamma,p})_+, \L^{\frac1M} \}.
\end{align}
\end{proposition}

 The proposition above can lead to the following
proposition.

\begin{proposition} \label{propZS}
If $\L$ satisfies the Lax equations then we have the following
Zakharov-Shabat equations
\begin{equation}
\label{zs}  (\A_{\alpha,m})_{t_{\beta,n}} - (\A_{\beta,
n})_{t_{\alpha,m}} + \{ \A_{\alpha,m} , \A_{\beta,n} \} =0
\end{equation}
for $-M+1 \leq \alpha, \beta \leq N$ ,  $m, n \geq 0$.
\end{proposition}

 Using the Zakharov-Shabat eqs.(\ref{zs})  the
flows of eqs.\eqref{edef2} can be proved to commute pairwise.
\begin{lemma} \label{zs corollary}
The following Zakharov-Shabat identities hold
\begin{align}\label{compatible zero curvature}
 &\partial_{\beta,n}(\B_{\alpha,m})_{-} - \partial_{\alpha,m}(\B_{\beta,
n})_{-} - \{(\B_{\alpha,m})_{-} , (\B_{\beta,n})_{-} \} =
0\\\label{compatible zero curvature+}
 &-\partial_{\beta,n}(\B_{\alpha,m})_{+} + \partial_{\alpha,m}(\B_{\beta,
n})_{+} - \{ (\B_{\alpha,m})_{+} , (\B_{\beta,n})_{+}\} =0
\end{align}
here, $-M+1\leq\alpha,\beta\leq N$,  $m,n \geq 0$.
\end{lemma}

Then following proposition will appear.

\begin{proposition} There exists $\phi = \phi(t,x)$ which is
characterized by
\begin{equation}\label{phi def}
    d\phi =  \sum_{\gamma=-M+1}^N \sum_{n=1}^\infty Res( \B_{\gamma,n} d\log k) dt_{\gamma,n}
           +\frac1M \log u_{-M} dx,
\end{equation}

where ``$d$'' means total differentiation in $(t,x)$, and $d\log k =
dk/k$. Furthermore $\phi$ satisfies the well-known dispersionless
(long-wave) limit of the two-dimensional Toda field equation
\begin{equation}\label{2-D toda field equation}  \d_{t_{-M+1,0}}\d_{t_{N,0}}\phi +\d_x
\exp( \d_x\phi)  =0.
\end{equation}
\end{proposition}
\begin{proof} The equation \eqref{phi def} is a compact form of the following
system
\begin{eqnarray}\label{dis phi flow}
    \d_{t_{\alpha,n}}  \phi& =& (\B_{\alpha,n} )_0,\\
\label{dis phi x flow2}
    \d_x \phi& = &\frac1M\log u_{-M}.
\end{eqnarray}

Taking the projection of \eqref{compatible zero curvature+} to $k^0$ term will lead to following identity
$$-\partial_{\beta,n}(\B_{\alpha,m})_0 +
\partial_{\alpha,m}(\B_{\beta, n})_0 - \{ (\B_{\alpha,m})_{+} ,
(\B_{\beta,n})_{+}\}_0 =0$$¡£
 Because $\{ (\B_{\alpha,m})_{+} ,
(\B_{\beta,n})_{+}\}_0=0$, we get
$$
\partial_{\alpha,m}(\B_{\beta, n})_0 =\partial_{\beta,n}(\B_{\alpha,m})_0 $$
i.e. the $t_{\alpha,m}$ flow and $t_{\beta,n}$ flow of \eqref{dis
phi flow} are compatible. Now we can see that the solution $\phi$ of
\eqref{phi def} exists. From \eqref{edef2}, we consider the $k^{-M}$
part.
\begin{eqnarray}
\frac{\d u_{-M}}{\d_{t_{\alpha,m}}}&=&[k(\frac{\d
(\B_{\alpha,m})_+}{\d k}\frac{\d \L}{\d x}-\frac{\d
(\B_{\alpha,m})_+}{\d x}\frac{\d \L}{\d
k})]|_{k^{-M}}\\
&=&M\frac{\d (\B_{\alpha,m})_0}{\d x}u_{-M},
\end{eqnarray}
which is
\begin{eqnarray}
\frac{\d \log u_{-M}}{\d_{t_{\alpha,m}}}&=&M\frac{\d
(\B_{\alpha,m})_0}{\d x},
\end{eqnarray}
i.e. eq. \eqref{dis phi flow} and eq. \eqref{dis phi x flow2} are
compatible. Considering a special case of \eqref{zs}
\begin{equation}
\label{zs special}  (A_{N,0})_{t_{-M+1,0}} - (A_{-M+1,0})_{t_{N,0}}
+ \{ A_{N,0} , A_{-M+1,0} \} =0
\end{equation}
 whose $k^0$ part is
\begin{equation}
\label{zs special'}  \d_{t_{-M+1,0}}(A_{N,0})_0 -\d
_{t_{N,0}}(A_{-M+1,0})_0 + \{ A_{N,0} , A_{-M+1,0} \}_0 =0.
\end{equation}
So we can get
\begin{equation}
\label{zs special'}  \d_{t_{-M+1,0}}(\B_{N,0})_0  + \{( \L^{\frac1N
})_+ , - ( \L^{\frac{1}M })_- \}_0 =0,
\end{equation}
which implies

\begin{equation}
\label{zs special''}  \d_{t_{-M+1,0}}\d_{t_{N,0}}\phi +\d_x
(u_{-M})^{\frac1M}  =0,
\end{equation}
i.e.
\begin{equation}\label{2-D toda field equation}  \d_{t_{-M+1,0}}\d_{t_{N,0}}\phi +\d_x
\exp( \d_x\phi)  =0.
\end{equation}
This is the end of proof.
\end{proof}
Eq.\eqref{2-D toda field equation} is just the dispersionless limit
of the generalized two-dimensional Toda field equation.

\begin{proposition}$\L$ is the Lax function of the
dispersionless BTH if and only if  there exists two Laurent series $\varphi_L$
$\varphi_R$ ({\it dressing fucntion}) which satisfies the equations
\begin{eqnarray}\label{sato dis}
    \nabla_{t_{\gamma,n},\varphi_L} \varphi_L& =& -(\B_{\gamma,n})_-,\quad
     \nabla_{t_{\gamma,n},\varphi_R} \varphi_R =  (\B_{\gamma ,n})_+,
\end{eqnarray}
where $-M+1\leq\gamma \leq N, n\geq 0.$
 $\varphi_L$ and
$\varphi_R$ have the following form
\begin{eqnarray} &&
\varphi_L=w_1k^{-1}+w_2k^{-2}+\ldots,
\label{dressP}\\
&& \varphi_R=\tilde{w_1}k+\tilde{w_2}k^2+ \ldots, \label{dressQ}
\end{eqnarray} where
$$
ad\varphi(\psi) = \{\varphi, \psi\}, \qquad \nabla_{t_{\gamma ,n},
\psi} \varphi = \sum_{m=0}^\infty \frac1{(m+1)!} (ad\psi)^m \left(
\frac{\d \varphi}{\d  t_{\gamma ,n}} \right).
$$
Such Laurent series $\varphi_L$ $\varphi_R$ are  unique up to transformation $\varphi_L \mapsto
H(\varphi_L,\psi_L)$, $\varphi_R \mapsto
H(\varphi_R,\psi_R)$, with a constant Laurent series $\psi_L =
\sum_{n=1}^\infty \psi_{Ln} k^{-n}$ ($\psi_{Ln}$: constant), $\psi_R =
\sum_{n=1}^\infty \psi_{Rn} k^{n}$ ($\psi_{Rn}$: constant) respectively,
where $H(X,Y)$ is the Hausdorff series  defined by
$$
    \exp(ad H(\varphi,\psi)) = \exp(ad\varphi) \exp(ad\psi).
$$

\end{proposition}
\begin{proof}
The proof is standard and similar as proof in \cite{Takasaki}. So we
omit it here.
\end{proof}
With the above preparation, in the next section we will consider the
Block type additional symmetry of the dBTH.

\section{Orlov-Schulman's $\M_L$, $\M_R$ functions}
To introduce the additional symmetry of the dBTH, We firstly define the following Orlov-Schulman's $\M_L$ functions as
\begin{eqnarray}\label{Moperator}
&&\M_L=e^{ad \varphi_L}(\Gamma_L)=e^{ad \varphi_L}e^{ad
t_L(k)}(\frac{x}{N}k^{-N}),
\end{eqnarray}
where
\begin{eqnarray}\Gamma_L=
\frac{x}{N}k^{-N}+\sum_{n\geq 0}\sum_
{\alpha=1}^{N}(n+1-\frac{\alpha-1}{N})
k^{N({n-\frac{\alpha-1}{N}})}t_{\alpha, n},
\end{eqnarray}
\begin{eqnarray}
t_L(k)=\sum_{n\geq 0}\sum_
{\alpha=1}^{N}k^{({n+1-\frac{\alpha-1}{N}})}t_{\alpha, n}.
\end{eqnarray}
$\M_L$ can be written in another form as following
\begin{eqnarray}\label{Moperator21'}
&&\M_L= \frac{x}{N}\L_L^{-1}+\sum_ {m=0}^{\infty}\sum_
{\gamma=1}^{N}v_{\gamma,m}(t,x)\L^{-(m+2+\frac {1-\gamma}
{N})}+\sum_{n\geq 0}\sum_ {\alpha=1}^{N}(n+1-\frac{\alpha-1}{N})
\L^{n-\frac{\alpha-1}{N}}t_{\alpha, n},
\end{eqnarray}
where $\L_L^{-1}:=e^{ad \varphi_L}(k^{-N})$.
 Similarly we define the
following Orlov-Schulman's $\M_R$ functions as following
\begin{eqnarray}\label{Moperator'}
&&\M_R=e^{ad \bar\varphi_R}(\Gamma_R)=e^{ad \phi}e^{ad
\varphi_R}(\Gamma_R)=e^{ad \phi} e^{ad \varphi_R}e^{ad
t_R(k^{-1})}(-\frac{x}{M}k^M) ,
\end{eqnarray}
where
\begin{eqnarray} \Gamma_R=-\frac{x}{M}k^M-\sum_{n\geq 0}\sum_
{\beta=-M+1}^{0}
({n+1+\frac{\beta}{M}})k^{-M({n+\frac{\beta}{M}})}t_{\beta, n},
\end{eqnarray}
\begin{eqnarray}
t_R(k^{-1})=\sum_{n\geq 0}\sum_
{\beta=-M+1}^{0}k^{-M(n+1+\frac{\beta}{M})}t_{\beta,
n}.
\end{eqnarray}

$\M_R$ can be written as following form
\begin{eqnarray}\label{Moperator22'}
&&\M_R=-\frac{x}{M}\L_R^{-1}+\sum_ {m=0}^{\infty}\sum_
{\gamma=-m+1}^{0}\bar v_{\gamma,m}(t,x)\L^{-(m+2+\frac {\gamma}
{M})}-\sum_{n\geq 0}\sum_
{\beta=-M+1}^{0}({n+1+\frac{\beta}{M}})\L^{n+\frac{\beta}{M}}t_{\beta,
n},
\end{eqnarray}
where $\L_R^{-1}:=e^{ad \bar\varphi_R}(k^{M})$.

 A direct calculation
shows that the Orlov-Schulman's functions  satisfies the following
theorem.
\begin{theorem}\label{flowsofM}
The following identities hold
\begin{eqnarray}
&\{\L,\M_L\}=1,\{\L,\M_R\}=1,\\
\label{alphaML} &\partial_{ t_{\gamma,n}}\M_L=
\{\A_{\gamma,n},\M_L\},\ \ \ &\partial_{
t_{\gamma,n}}\M_R=\{\A_{\gamma,n},\M_R\},\\
\label{M_L^nL^k}&\dfrac{\partial
\M_L^n\L^k}{\partial{t_{\gamma,n}}}=\{\A_{\gamma,n}, \M_L^n\L^k\},\
\ \
 &\dfrac{\partial
\M_R^n\L^k}{\partial{t_{\gamma,n}}}=\{\A_{\gamma,n}, \M_R^n\L^k\},
\end{eqnarray}
where $1-M\leq \gamma \leq N.$
\end{theorem}
\begin{proof}
Before the proof, we  firstly set $1\leq \alpha \leq N, -M+1\leq
\beta \leq 0$. Then the following calculation will lead to one part
of the first equation of eq.\eqref{alphaML}
\begin{eqnarray*}
\partial_{ t_{\alpha,n}}\M_L&=&
\partial_{ t_{\alpha,n}}e^{ad \varphi_L}(\Gamma_L)\\
&=&e^{ad \varphi_L}\partial_{
t_{\alpha,n}}(\Gamma_L)+\{\nabla_{t_{\alpha,n}, \varphi_L}
\varphi_L,\M_L \}\\
&=& e^{ad \varphi_L}(\sum_{n\geq 0}\sum_
{\alpha=1}^{N}(n+1-\frac{\alpha-1}{N} )
k^{N({n-\frac{\alpha-1}{N}})})+\{-(\B_{\alpha,n})_{-},\M_L \}\\
&=&(n+1-\frac{\alpha-1}{N}
)\L^{{n-\frac{\alpha-1}{N}}}+\{-(\B_{\alpha,n})_{-},\M_L \}\\
&=&\{\L^{{n+1-\frac{\alpha-1}{N}}},\M_L\}+\{-(\B_{\alpha,n})_{-},\M_L \}\\
&=& \{\B_{\alpha,n},\M_L\}+\{-(\B_{\alpha,n})_{-},\M_L \}\\
&=&  \{(\B_{\alpha,n})_{+},\M_L\}
\end{eqnarray*}
Similarly $t_{\beta,n}$ flow of $\M_L$ is as following
\begin{eqnarray*}
\partial_{ t_{\beta,n}}\M_L&=&
\partial_{ t_{\beta,n}}e^{ad \varphi_L}(\Gamma_L)\\
&=&e^{ad \varphi_L}\partial_{
t_{\beta,n}}(\Gamma_L)+\{\nabla_{t_{\beta,n}, \varphi_L}
\varphi_L,\M_L \}\\
&=& \{-(\B_{\beta,n})_{-},\M_L \},
\end{eqnarray*}
which imply the other part of the first equation of
eq.\eqref{alphaML}.

In the same way, by calculation we can prove  results
\begin{eqnarray}
\partial_{ t_{\alpha,n}}\M_R&=& \{(\B_{\alpha,n})_{+},\M_R\}
\end{eqnarray}
and
\begin{eqnarray}
\partial_{ t_{\beta,n}}\M_R
&=&  \{-(\B_{\beta,m})_{-},\M_R\}.
\end{eqnarray}
Till now we have finished the proof of eq.\eqref{alphaML}. By
eq.\eqref{alphaML} and eq.\eqref{edef2}, we can prove
eq.\eqref{M_L^nL^k} easily.
\end{proof}
We can formulate the following 2-form
\begin{eqnarray}\label{omega}
\omega&=&\frac{d k}k\wedge dx+\sum_{\alpha=-M+1}^{N}\sum_{n\geq 0}d
\A_{\alpha,n}\wedge dt_{\alpha,n}
\end{eqnarray}
which satisfies
\begin{equation}
d\omega=0,
\end{equation}
and the following proposition.
\begin{proposition}The identity
\begin{equation}
\omega\wedge\omega=0
\end{equation}
 is equivalent to the Zakharov-Shabat equations.
\end{proposition}
\begin{proof}
Eq.\eqref{omega} can have the following detailed representation
\begin{eqnarray*}
\omega=\frac{d k}k\wedge
dx+\sum_{\alpha,\beta=-M+1}^{N}\sum_{m,n\geq 0}\frac{\d
\A_{\alpha,n}}{\d t_{\beta,m}}d t_{\beta,m}\wedge
dt_{\alpha,n}+\sum_{\alpha=-M+1}^{N}\sum_{n\geq 0}(\frac{\d
\A_{\alpha,n}}{\d k}d k+\frac{\d \A_{\alpha,n}}{\d x}d x)\wedge
dt_{\alpha,n}.
\end{eqnarray*}
We can construct wedge product $\omega\wedge\omega$ as following
\begin{eqnarray*}
&&\omega\wedge\omega\\
&=&(\frac{d k}k\wedge
dx+\sum_{\alpha,\beta=-M+1}^{N}\sum_{m,n\geq 0}\frac{\d
\A_{\alpha,n}}{\d t_{\beta,m}}d t_{\beta,m}\wedge
dt_{\alpha,n}+\sum_{\alpha=-M+1}^{N}\sum_{n\geq 0}(\frac{\d
\A_{\alpha,n}}{\d k}d k+\frac{\d \A_{\alpha,n}}{\d x}d x)\wedge
dt_{\alpha,n})\\
&&\wedge(\frac{d k}k\wedge
dx+\sum_{\alpha,\beta=-M+1}^{N}\sum_{m,n\geq 0}\frac{\d
\A_{\alpha,n}}{\d t_{\beta,m}}d t_{\beta,m}\wedge
dt_{\alpha,n}+\sum_{\alpha=-M+1}^{N}\sum_{n\geq 0}(\frac{\d
\A_{\alpha,n}}{\d k}d k+\frac{\d \A_{\alpha,n}}{\d x}d x)\wedge
dt_{\alpha,n})\\
&=&\frac2k \sum_{\alpha,\beta=-M+1}^{N}\sum_{m,n\geq 0}(\frac{\d
\A_{\alpha,n}}{\d t_{\beta,m}}-\frac{\d \A_{\beta,m}}{\d
t_{\alpha,n}})d k\wedge dx\wedge d t_{\beta,m}\wedge dt_{\alpha,n}\\
&&+ 2\sum_{\alpha,\beta=-M+1}^{N}\sum_{m,n\geq 0}(\frac{\d
\A_{\alpha,n}}{\d k}\frac{\d \A_{\beta,m}}{\d x}-\frac{\d
\A_{\beta,m}}{\d k}\frac{\A_{\alpha,n}}{\d x})d k\wedge dx\wedge d
t_{\beta,m}\wedge d t_{\alpha,n}\\
&&+ 2\sum_{\alpha,\beta,\gamma=-M+1}^{N}\sum_{m,n\geq 0}[(\frac{\d
\A_{\alpha,n}}{\d t_{\beta,m}}-\frac{\d \A_{\beta,m}}{\d
t_{\alpha,n}})\frac{\d \A_{\gamma,l}}{\d k}-(\frac{\d
\A_{\gamma,l}}{\d t_{\beta,m}}-\frac{\d \A_{\beta,m}}{\d
t_{\gamma,l}})\frac{\d \A_{\alpha,n}}{\d k}\\
&&+(\frac{\d
\A_{\gamma,l}}{\d t_{\alpha,n}}-\frac{\d \A_{\alpha,n}}{\d t_{\gamma,l}})\frac{\d \A_{\beta,m}}{\d k}]d
k\wedge d t_{\beta,m}\wedge dt_{\alpha,n}\wedge dt_{\gamma,l}\\
&&+ 2\sum_{\alpha,\beta,\gamma=-M+1}^{N}\sum_{m,n\geq 0}[(\frac{\d
\A_{\alpha,n}}{\d t_{\beta,m}}-\frac{\d \A_{\beta,m}}{\d
t_{\alpha,n}})\frac{\d \A_{\gamma,l}}{\d x}-(\frac{\d
\A_{\gamma,l}}{\d t_{\beta,m}}-\frac{\d \A_{\beta,m}}{\d
t_{\gamma,l}})\frac{\d \A_{\alpha,n}}{\d x}\\
&&+(\frac{\d \A_{\alpha,n}}{\d t_{\gamma,l}}-\frac{\d
\A_{\gamma,l}}{\d t_{\alpha,n}})\frac{\d \A_{\beta,m}}{\d x}]d
x\wedge d t_{\beta,m}\wedge dt_{\alpha,n}\wedge dt_{\gamma,l}.
\end{eqnarray*}
So $\omega\wedge\omega=0$ is equivalent to $\frac{\d
\A_{\alpha,n}}{\d t_{\beta,m}}-\frac{\d \A_{\beta,m}}{\d
t_{\alpha,n}}+k(\frac{\d \A_{\alpha,n}}{\d k}\frac{\d
\A_{\beta,m}}{\d x}-\frac{\d \A_{\beta,m}}{\d
k}\frac{\A_{\alpha,n}}{\d x})=0$, i.e. $\frac{\d \A_{\alpha,n}}{\d
t_{\beta,m}}-\frac{\d \A_{\beta,m}}{\d
t_{\alpha,n}}+\{\A_{\alpha,n},\A_{\beta,m}\}=0$ which is just
eq.\eqref{zs}.
\end{proof}
Similarly we  get the following corollary.

\begin{proposition}
 The dBTH is equivalent to the following exterior differential
equations
\begin{equation}\label{wedge}
d\L\wedge d\M_L= d\L\wedge d\M_R=\omega
\end{equation}
\end{proposition}
\begin{proof}
By expanding the left side of identity
\begin{equation}
d\L\wedge d\M_L=\omega,
\end{equation}
the following result can be got.
\begin{eqnarray*}
d\L\wedge d\M_L&=&(\frac{\d\L}{\d k}d k+\frac{\d\L}{\d x}d
x+\sum_{\alpha=-M+1}^{N}\sum_{n\geq 0}\frac{\d \L}{\d t_{\alpha,n}}d
t_{\alpha,n})\wedge\\
&&(\frac{\d\M_L}{\d k}d k+\frac{\d\M_L}{\d x}d
x+\sum_{\alpha=-M+1}^{N}\sum_{n\geq 0}\frac{\d \M_L}{\d
t_{\alpha,n}}d t_{\alpha,n})\\
&=&(\frac{d k}k\wedge dx+\sum_{\alpha,\beta=-M+1}^{N}\sum_{m,n\geq
0}\frac{\d \A_{\alpha,n}}{\d t_{\beta,m}}d t_{\beta,m}\wedge
dt_{\alpha,n}+\sum_{\alpha=-M+1}^{N}\sum_{n\geq 0}(\frac{\d
\A_{\alpha,n}}{\d k}d k\\
&&+\frac{\d \A_{\alpha,n}}{\d x}d x)\wedge dt_{\alpha,n}).
\end{eqnarray*}
By comparing the coefficients of $d k\wedge dx$, we get the
canonical relation $\{\L,\M_L\}=1$. By comparing the coefficients of
$d k\wedge  dt_{\alpha,n}, d x\wedge  dt_{\alpha,n},d
t_{\beta,m}\wedge dt_{\alpha,n}$, we get the canonical relation
\begin{eqnarray}\label{1}\frac{\d\L}{\d k}\frac{\d \M_L}{\d t_{\alpha,n}}-\frac{\d\M_L}{\d
k}\frac{\d \L}{\d t_{\alpha,n}}&=&\frac{\d \A_{\alpha,n}}{\d
k}\\\label{2} \frac{\d\L}{\d x}\frac{\d \M_L}{\d
t_{\alpha,n}}-\frac{\d\M_L}{\d x}\frac{\d \L}{\d
t_{\alpha,n}}&=&\frac{\d \A_{\alpha,n}}{\d x}\\\label{3}
\frac{\d\L}{\d t_{\beta,m}}\frac{\d \M_L}{\d
t_{\alpha,n}}-\frac{\d\M_L}{\d t_{\beta,m}}\frac{\d \L}{\d
t_{\alpha,n}}&=&\frac{\d \A_{\alpha,n}}{\d t_{\beta,m}}-\frac{\d
\A_{\beta,m}}{\d t_{\alpha,n}},
\end{eqnarray}
eq.\eqref{1} and eq.\eqref{2} imply eq.\eqref{edef2},
eq.\eqref{alphaML}. Eq.\eqref{edef2}, eq.\eqref{alphaML} can lead to
eq.\eqref{3} as following.
\begin{eqnarray*}
&&\frac{\d\L}{\d t_{\beta,m}}\frac{\d \M_L}{\d
t_{\alpha,n}}-\frac{\d\M_L}{\d t_{\beta,m}}\frac{\d \L}{\d
t_{\alpha,n}}\\&=&
\{\A_{\beta,m},\L\}\{\A_{\alpha,n},\M_L\}-\{\A_{\beta,m},\M_L\}\{\A_{\alpha,n},\L\}\\
&=& k^2[(\frac{\d ( \B_{\beta,m})_+}{\d k}\frac{\d \L}{\d
x}-\frac{\d \A_{\beta,m}}{\d x}\frac{\d \L}{\d k})(\frac{\d
\A_{\alpha,n}}{\d k}\frac{\d \M_L}{\d x}-\frac{\d
\A_{\alpha,n}}{\d x}\frac{\d \M_L}{\d k})\\
&&-(\frac{\d \A_{\beta,m}}{\d k}\frac{\d \M_L}{\d x}-\frac{\d
\A_{\beta,m}}{\d x}\frac{\d \M_L}{\d k})(\frac{\d \A_{\alpha,n}}{\d
k}\frac{\d \L}{\d x}-\frac{\d
\A_{\alpha,n}}{\d x}\frac{\d \L}{\d k})]\\
&=& k^2[\frac{\d \A_{\beta,m}}{\d k}\frac{\d \M_L}{\d x}\frac{\d
\A_{\alpha,n}}{\d x}\frac{\d \L}{\d k}-\frac{\d \A_{\beta,m}}{\d
k}\frac{\d \L}{\d x}\frac{\d
\A_{\alpha,n}}{\d x}\frac{\d \M_L}{\d k}\\
&&+\frac{\d \A_{\beta,m}}{\d x}\frac{\d \M_L}{\d k}\frac{\d
\A_{\alpha,n}}{\d k}\frac{\d \L}{\d x}-\frac{\d \A_{\beta,m}}{\d
x}\frac{\d \L}{\d k}\frac{\d
\A_{\alpha,n}}{\d k}\frac{\d \M_L}{\d x}]\\
&=& k(\frac{\d \A_{\beta,m}}{\d k}\frac{\d \A_{\alpha,n}}{\d x}
-\frac{\d \A_{\beta,m}}{\d x}\frac{\d \A_{\alpha,n}}{\d
k})\\
&=& \{\A_{\beta,m}, \A_{\alpha,n}\}\\
&=&\frac{\d \A_{\alpha,n}}{\d t_{\beta,m}}-\frac{\d
\A_{\beta,m}}{\d t_{\alpha,n}}.
\end{eqnarray*}
 In the same way,
we can get all the equations for $\M_R$ from the second identity in eq.\eqref{wedge}.
\end{proof}

Using equation \eqref{Moperator21'} and \eqref{Moperator22'}, taking
derivatives of them will lead to following lemma.
\begin{lemma}\label{vderivative}
Following formula will hold
\begin{eqnarray}
\frac{\d v_{\gamma_1,n}(t,x)}{\d t_{\gamma_2,m}} &=&\res \B_{\gamma_1,n}d_k \A_{\gamma_2,m},\\
\frac{\d \bar v_{\gamma_1,n}(t,x)}{\d t_{\gamma_2,m}} &=&\res
\B_{\gamma_1,n}d_k
\A_{\gamma_2,m},\\
\frac{\d v_{\gamma,n}(t,x)}{\d x} &=&\res \B_{\gamma,n} d \log k,\\
\frac{\d \bar v_{\gamma,n}(t,x)}{\d x} &=&\res \B_{\gamma,n} d \log
k,
\end{eqnarray}
where $ -M+1\leq\gamma,\gamma_1,\gamma_2\leq N$.
\end{lemma}
\begin{proof}
Firstly we take derivatives of $\M_L$ by $t_{\alpha,n}$ and get
following calculation,
\begin{eqnarray*}
\frac{\d\M_L}{\d t_{\alpha,n}}&=&\frac{\d\M_L}{\d \L}\frac{\d\L}{\d
t_{\alpha,n}}+\sum_ {m=0}^{\infty}\sum_ {\gamma=1}^{N}\frac{\d
v_{\gamma,m}(t,x)}{\d t_{\alpha,n}}\L^{-(m+2+\frac {1-\gamma}
{N})}+\sum_{n\geq 0}\sum_ {\alpha=1}^{N}(n+1+\frac {1-\gamma}
{N})\B_{\alpha,n-1}
\end{eqnarray*}
 where $v_{1,0}=\frac x N$,
  which leads to
\begin{eqnarray*}
&&\frac{\d v_{\gamma,m}(t,x)}{\d t_{\alpha,n}}\\
&=&\res \L^{m+1-\frac {\gamma-1} {N}}(\frac{\d\M_L}{\d
t_{\alpha,n}}d_k\L-\frac{\d\M_L}{\d \L}\frac{\d\L}{\d
t_{\alpha,n}}d_k\L- \sum_{n\geq 0}\sum_ {\alpha=1}^{N}(n+1+\frac
{1-\gamma}
{N})\B_{\alpha,n-1} d_k\L)\\
&=&\res \B_{\gamma,m}(\frac{\d\M_L}{\d
t_{\alpha,n}}d_k\L-\frac{\d\M_L}{\d \L}\frac{\d\L}{\d
t_{\alpha,n}}d_k\L)
\\
&=&\res \B_{\gamma,m}k[(\frac{\d \A_{\alpha,n}}{\d
k}\frac{\d\M_L}{\d x}-\frac{\d \A_{\alpha,n}}{\d x}\frac{\d\M_L}{\d
k})d_k\L-(\frac{\d \A_{\alpha,n}}{\d k}\frac{\d\L}{\d x}-\frac{\d
\A_{\alpha,n}}{\d x}\frac{\d\L}{\d k})d_k\M_L]\\
&=&\res\B_{\gamma,m}\frac{\d \A_{\alpha,n}}{\d k}\{\L,\M_L\}\\
&=&\res \B_{\gamma,m}d_k \A_{\alpha,n}
\end{eqnarray*}

Now  we take derivatives of $\M_L$ by $x$ and get following
calculation,

\begin{eqnarray}
\frac{\d\M_L}{\d x}=\frac{\d\M_L}{\d \L}\frac{\d\L}{\d x}+\sum_
{m=0}^{\infty}\sum_ {\gamma=1}^{N}\frac{\d v_{\gamma,m}(t,x)}{\d
x}\L^{-(m+2+\frac {1-\gamma} {N})},
\end{eqnarray}

which leads to
\begin{eqnarray*}
&&\frac{\d v_{\gamma,m}(t,x)}{\d x}\\
&=&\res\L^{m+1-\frac {\gamma-1} {N}}(\frac{\d\M_L}{\d
x}d_k\L-\frac{\d\M_L}{\d
\L}\frac{\d\L}{\d x}d_k\L)\\
&=&\res  \B_{\gamma,m}(\frac{\d\M_L}{\d x}d_k\L-\frac{\d\L}{\d
x}d_k\M_L)
\\
&=&\res  \B_{\gamma,m} d \log k.
\end{eqnarray*}
The other cases can be proven in similar ways.
\end{proof}
By the Lemma above, it is time to introduce $S$ functions which is
included in the following proposition.
\begin{proposition}
There exist functions $S_L$ and $S_R$ which satisfy
\begin{eqnarray}\label{dS_L}
dS_L&=&\M_Ld\L+\log k  dx+\sum_{\alpha=-M+1}^{N}\sum_{n\geq 0}
\A_{\alpha,n}\wedge dt_{\alpha,n},\\\label{dS_R} dS_R&=&\M_Rd\L+\log
k  dx+\sum_{\beta=-M+1}^{N}\sum_{n\geq 0} \A_{\beta,n}\wedge
dt_{\beta,n},
\end{eqnarray}
where $S_L,S_R$ have the following Laurent expansion.
\begin{eqnarray}\label{Moperator3}
S_L= \frac{x}{N}\log_+\L-\sum_ {m=0}^{\infty}\sum_
{\gamma=1}^{N}v_{\gamma,m}(t,x)\frac{ \L^{-(m+1+\frac {1-\gamma}
{N})}}{m+1+\frac {1-\gamma} {N}}+\sum_{n\geq 0}\sum_
{\alpha=1}^{N}\B_{\alpha,n}t_{\alpha, n},
\end{eqnarray}
\begin{eqnarray}
S_R=-\frac{x}{M}\log_-\L-\sum_ {m=0}^{\infty}\sum_
{\gamma=-m+1}^{0}\bar v_{\gamma,m}(t,x)\frac{\L^{-(m+1+\frac
{\gamma} {M})}}{m+1+\frac {\gamma} {M}} -\sum_{n\geq 0}\sum_
{\beta=-M+1}^{0}\B_{\beta,n}t_{\beta, n}.
\end{eqnarray}
\end{proposition}
\begin{proof}
We can prove the right hand sides of eq.\eqref{dS_L} and
eq.\eqref{dS_R} are closed according to eq.\eqref{wedge}. That
implies the existence of $S_L, S_R$.

By Lemma \ref{vderivative}, we can prove the form of $S_L$ is
correct using following computation,
\begin{eqnarray*}
&&\frac{\d S_L}{\d
t_{\alpha,n}}|_{\L,t_{\beta,m}(\beta\neq\alpha,\  or\ n\neq
m),x}\\
&=&-\sum_ {m=0}^{\infty}\sum_ {\gamma=1}^{N}\frac{
\d v_{\gamma,m}(t,x)}{\d t_{\alpha,n}}\frac{ \L^{-(m+1+\frac
{1-\gamma} {N})}}{m+1+\frac {1-\gamma} {N}}+\B_{\alpha,n}\\
&=&-\sum_ {m=0}^{\infty}\sum_ {\gamma=1}^{N}(\res \B_{\gamma,m}d_k
\A_{\alpha,n})\frac{ \L^{-(m+1+\frac {1-\gamma}
{N})}}{m+1+\frac {1-\gamma} {N}}+\B_{\alpha,n}\\
&=&\sum_ {m=0}^{\infty}\sum_ {\gamma=1}^{N}(\res
\A_{\alpha,n}d_k\B_{\gamma,m} )\frac{ \L^{-(m+1+\frac {1-\gamma}
{N})}}{m+1+\frac {1-\gamma} {N}}+\B_{\alpha,n}\\
&=&\sum_ {m=0}^{\infty}\sum_ {\gamma=1}^{N}(\res
\A_{\alpha,n}\L^{m+\frac {1-\gamma} {N}}d_k\L ) \L^{-(m+1+\frac
{1-\gamma} {N})}+\B_{\alpha,n}\\
&=&-(\B_{\alpha,n})_-+\B_{\alpha,n}\\
&=&\A_{\alpha,n}.
\end{eqnarray*}

\end{proof}

Now, we  will define the tau function of dBTH in following
proposition.
\begin{proposition}
There exists a function $\tau_d(t,x)$ which satisfies
\begin{eqnarray}\label{d tau}
d \log \tau_d(t,x)=\sum_{n\geq 0}\sum_ {\alpha=1}^{N}v_{\alpha,n}d
t_{\alpha,n}+\sum_{n\geq 0}\sum_ {\beta=-M+1}^{0}\bar v_{\beta,n}d
t_{\beta,n}+\phi dx
\end{eqnarray}
\end{proposition}
\begin{proof}
To prove the existence of tau function, we need to prove the
compatibility of all the time flows which can be shown in following
calculation,
\begin{eqnarray*}
\frac{\d v_{\alpha_1,n_1}}{\partial{ t_{\alpha_2,n_2}}}&=&\res
\B_{\alpha_1,n_1}d_k (\B_{\alpha_2,n_2})_+\\
&=&\res
(\B_{\alpha_1,n_1})_-d_k (\B_{\alpha_2,n_2})_+\\
&=&\res
(\B_{\alpha_1,n_1})_-d_k \B_{\alpha_2,n_2}\\
&=&-\res
\B_{\alpha_2,n_2}d_k(\B_{\alpha_1,n_1})_- \\
&=&\res
\B_{\alpha_2,n_2}d_k (\B_{\alpha_1,n_1})_+\\
 &=&\frac{\d v_{\alpha_2,n_2}}{\partial{ t_{\alpha_1,n_1}}},\end{eqnarray*}
 \begin{eqnarray*}
\frac{\d v_{\alpha,n}}{\partial{ t_{\beta,m}}}&=&-\res
\B_{\alpha,n}d_k (\B_{\beta,m})_-\\
&=&-\res
(\B_{\alpha,n})_+d_k (\B_{\beta,m})_-\\
&=&-\res
(\B_{\alpha,n})_+d_k \B_{\beta,m}\\
&=&\res
\B_{\beta,m}d_k(\B_{\alpha,n})_+\\
 &=&\frac{\d \bar v_{\beta,m}}{\partial{ t_{\alpha,n}}},
 \end{eqnarray*}
 \begin{eqnarray*}
\frac{\d v_{\gamma,n}}{\partial x}&=&\res
\B_{\gamma,n}d \log k=
(\B_{\gamma,n})_0
 =\frac{\d \phi}{\partial{ t_{\gamma,n}}}.
 \end{eqnarray*}
 The other cases for commutativity can be proved in similar ways.
 So the 1-form of the right side of eq.\eqref{d tau} is closed.
 Therefore, there exist a tau function which satisfies  eq.\eqref{d
 tau}.
\end{proof}

\section{ Additional Symmetries of dBTH}
We are now in a position to define the additional flows, and then to
prove that they are symmetries, which are called additional
symmetries of the dBTH. We introduce additional
independent variables $t^*_{m,l}$, and define the action of the
additional flows on the wave operator as
\begin{equation}\label{definitionadditionalflowsonphi2d}
\nabla_{t^*_{m,l}, \varphi_L} \varphi_L
=-\left((\M_L-\M_R)^m\L^l\right)_{-},
\end{equation}

\begin{equation}\label{definitionadditionalflowsonphi2dr}
\nabla_{t^*_{m,l}, \varphi_R} \varphi_R
=\left((\M_L-\M_R)^m
\L^l\right)_{+},
\end{equation}

\begin{proposition}\label{additionalflowsLM}
The additional flows act on $\L$ and $\M_L$, $\M_R$as
\begin{equation}\label{ETHadditionalflow1}
\dfrac{\partial \L}{\partial{t^*_{m,l}}}=\{((\M_L-\M_R)^m
\L^l)_{+}, \L\},
\end{equation}
\begin{equation}\label{ETHadditionalflow2}
\dfrac{\partial \M_L}{\partial{t^*_{m,l}}}=-\{((\M_L-\M_R)^m
\L^l)_{-},
\M_L\}
\end{equation}
\begin{equation}\label{ETHadditionalflow3}
\dfrac{\partial \M_R}{\partial{t^*_{m,l}}}=\{((\M_L-\M_R)^m
\L^l)_{+},
\M_R\}
\end{equation}
\end{proposition}
\noindent \textbf{Proof} By performing the derivative on $\L$ and
using eq.(\ref{ETHadditionalflow1}), we get
\begin{eqnarray*}
\partial_{t^*_{m,l}}\L=&& \partial_{t^*_{m,l}}(e^{ad \varphi_L}(k^N))\\&
=&\{\nabla_{t^*_{m,l}, \varphi_L} \varphi_L,e^{ad \varphi_L}(k^N)\}\\
&=&-[((\M_L-\M_R)^m
\L^l)_{-}, \L].
\end{eqnarray*}
For the action on $\M_L$ and $\M_R$ given in eq.(\ref{Moperator}),
there exists similar derivation as $(\partial_{t^*_{m,l}}\L)$,i.e.
\begin{eqnarray*}
(\partial_{t^*_{m,l}}\M_L)&=&\partial_{t^*_{m,l}}(e^{ad \varphi_L}(\Gamma_L))\\
&=&\{\nabla_{t^*_{m,l}, \varphi_L} \varphi_L,e^{ad \varphi_L}(\Gamma_L)\}\\
&=&-\{((\M_L-\M_R)^m
\L^l)_{-}, \M_L\}.
\end{eqnarray*}
Here the fact that $\Gamma_L$ does not depend on the additional
flows variables $t^*_{m,l}$ has been used. Other identities can
also be got in the similar way. \hfill $\square$ \\

From that, we can prove the following corollary:
%%%%%%%%%%%%%%%%%%%%%%%%%%%%%%%%%%%%%%%%%%%%%%%%%%%%%
\begin{corollary}\label{additionflowsonLnMmAnk}
The following several equations hold
\begin{equation}\label{ETHadditionalflow4}
\dfrac{\partial \L^n}{\partial{t^*_{m,l}}}=\{((\M_L-\M_R)^m
\L^l)_{+},
\L^n\},\; \dfrac{\partial
\L^n}{\partial{t^*_{m,l}}}=\{-((\M_L-\M_R)^m
\L^l)_{-}, \L^n\}
\end{equation}
\begin{equation}\label{ETHadditionalflow4'}
\dfrac{\partial
\B_{\alpha,n}}{\partial{t^*_{m,l}}}=-\{((\M_L-\M_R)^m
\L^l)_{-},
\B_{\alpha,n}\},\; \dfrac{\partial
\B_{\beta,n}}{\partial{t^*_{m,l}}}=\{((\M_L-\M_R)^m
\L^l)_{+},
\B_{\beta,n}\}
\end{equation}
\begin{equation}\label{ETHadditionalflow5}
\dfrac{\partial \M_\L^n}{\partial{t^*_{m,l}}}=-\{((\M_L-\M_R)^m
\L^l)_{-},
\M_\L^n\},\;
 \dfrac{\partial\M_R^n}{\partial{t^*_{m,l}}}=\{((\M_L-\M_R)^m
\L^l)_{+}, \M_\L^n\},
\end{equation}

\begin{equation}\label{eqadditionflowsonLnMmAnk'}
 \dfrac{\partial \M_\L^n\L^k}{\partial{t^*_{m,l}}}=-\{((\M_L-\M_R)^m
\L^l)_{-}, \M_\L^n\L^k\},\;
\end{equation}
\begin{equation}\label{eqadditionflowsonLnMmAnk}
 \dfrac{\partial \M_R^n\L^k}{\partial{t^*_{m,l}}}=\{((\M_L-\M_R)^m
\L^l)_{+}, \M_R^n\L^k\}.
\end{equation}

\end{corollary}
%%%%%%%%%%%%%%%%%%%%%%%%%%%%%%%%%%%%%%%%%%%%%%%%%%%%%
\noindent \textbf{Proof} We present here only the proof of the first
equation. The others can be proved in a similar way. The derivative
of $\L^n$ with respect to $t^*_{m,l}$ leads to
\begin{eqnarray*}
\dfrac{\partial \L^n}{\partial{t^*_{m,l}}}=\dfrac{\partial
\L}{\partial{t^*_{m,l}}}\L^{n-1}+ \L \dfrac{\partial
\L}{\partial{t^*_{m,l}}} \L^{n-2}+\cdots +\L^{n-2} \dfrac{\partial
\L}{\partial{t^*_{m,l}}} \L +\L^{n-1} \dfrac{\partial
\L}{\partial{t^*_{m,l}}}=\sum\limits_{k=1}^n \L^{k-1} \dfrac{\partial
\L}{\partial{t^*_{m,l}}} \L^{n-k}
\end{eqnarray*}
and then taking $\dfrac{\partial
\L}{\partial{t_{m,l}}}=\{((\M_L-\M_R)^m
\L^l)_{+}, \L\}$ into the above
formula. After that, we get
\begin{eqnarray*}
\dfrac{\partial \L^n}{\partial{t^*_{m,l}}}=\sum\limits_{k=1}^n \L^{k-1}
\{((\M_L-\M_R)^m
\L^l)_{+}, \L\}
 \L^{n-k}=\{((\M_L-\M_R)^m
\L^l)_{+}, \L^n\}. \text{\hspace{5cm}} \square
\end{eqnarray*}
\begin{proposition}
The additional flows $\dfrac{\partial }{\partial{t^*_{m,l}}}$ commute with the dispersionless bigraded Toda
hierarchy flows $\dfrac{\partial }{\partial{t_{c,n}}}$, i.e.
\begin{equation}
[\partial_{t^*_{m,l}}, \partial_{t_{c,n}}]\L=0,
\end{equation}where $-M+1\leq c \leq N.$
Here
$\partial_{t^*_{m,l}}=\frac{\partial}{\partial{t^*_{m,l}}},
\partial_{t_{c,n}}=\frac{\partial}{\partial{t_{c,n}}}$.
\end{proposition}
\begin{proof}
According to the definition
and using the action of the additional flows  on $\L$, we get
\begin{eqnarray*}
&&[\partial_{t^*_{m,l}},\partial_{t_{c,n}}]\L\\
 &=&\partial_{t^*_{m,l}}\partial_{t_{c,n}}e^{ad \varphi_L}(k^N)-
 \partial_{t_{c,n}}\partial_{t^*_{m,l}}e^{ad \varphi_L}(k^N)\\
&=& \partial_{t^*_{m,l}}\{\nabla_{t_{c,n}, \varphi_L} \varphi_L,e^{ad
\varphi_L}(k^N)\}-
 \partial_{t_{c,n}}\{\nabla_{t^*_{m,l}, \varphi_L} \varphi_L, e^{ad \varphi_L}(k^N)\}\\
&=& \partial_{t^*_{m,l}}\{-(\B_{c,n})_-,\L\}-
 \partial_{t_{c,n}}\{\left((\M_L-\M_R)^m
\L^l\right)_{+}, \L\}\\
&=&
\{-\{\left((\M_L-\M_R)^m
\L^l\right)_{+},\B_{c,n}\}_-,\L\}+\{-(\B_{c,n})_-,\{\left((\M_L-\M_R)^m
\L^l\right)_{+},
\L\}\}- \\
&&\{\{-(\B_{c,n})_-,(\M_L-\M_R)^m
\L^l\}_{+}, \L\}- \{\left((\M_L-\M_R)^m
\L^l\right)_{+}, \{-(\B_{c,n})_-,\L\}\}\\
&=&
\{-\{\left((\M_L-\M_R)^m
\L^l\right)_{+},\B_{c,n}\}_-,\L\}+\{\L,\{(\B_{c,n})_-,\left((\M_L-\M_R)^m
\L^l\right)_{+}
\}\}\\
&&- \{\{-(\B_{c,n})_-,(\M_L-\M_R)^m
\L^l\}_{+}, \L\}
\\
&=&\{\L,\{(\B_{c,n})_-,\left((\M_L-\M_R)^m
\L^l\right)_{+}
\}+\{\left((\M_L-\M_R)^m
\L^l\right)_{+},\B_{c,n}\}_-+\\
&&\{-(\B_{c,n})_-,(\M_L-\M_R)\L^l\}_{+}\}\\
&=&0.
\end{eqnarray*}
\end{proof}
The commutativity of additional flows with flows of the dBTH means that the additional flows are symmetries of the dBTH. It is a kind of Block type symmetry of the dBTH which will be shown in the next proposition.

\begin{proposition}\label{WinfiniteCalgebra}
Additional flows  $\partial_{t^*_{m,l}}(m\geq 0,l\geq 0)$ form the
following Block type Lie algebra
 \begin{eqnarray}\label{algebra relation}
[\partial_{t^*_{m,l}},\partial_{t^*_{n,k}}]\L= (km-nl)\d^*_{m+n-1,k+l-1}\L.
\end{eqnarray}
where  $m,n\geq 0; l,k\geq 0.$
\end{proposition}
\begin{proof}
 By using
Proposition\ref{additionalflowsLM}, we get
\begin{eqnarray*}
[\partial_{t^*_{m,l}},\partial_{t^*_{n,k}}]\L&=&
\partial_{t^*_{m,l}}(\partial_{t^*_{n,k}}\L)-
\partial_{t^*_{n,k}}(\partial_{t^*_{m,l}}\L)\\
&=&-\partial_{t^*_{m,l}}\{((\M_L-\M_R)^n\L^k)_{-},\L\}
+\partial_{t^*_{n,k}}\{((\M_L-\M_R)^m\L^l)_{-},\L\}\\
&=&-\{\{\partial_{t^*_{m,l}}
(\M_L-\M_R)^n\L^k\}_{-},\L\}-\{((\M_L-\M_R)^n\L^k)_{-},\partial_{t^*_{m,l}} \L\}\\
&&+ \{\{\partial_{t^*_{n,k}} (\M_L-\M_R)^m\L^l)_{-},\L\}+\{
((\M_L-\M_R)^m\L^l)_{-},\partial_{t^*_{n,k}} \L\},
\end{eqnarray*}
which further leads to the following calculation
 \begin{eqnarray*}&&
[\partial_{t^*_{m,l}},\partial_{t^*_{n,k}}]\L\\
&=&-\left\{\Big[\sum_{p=0}^{n-1}
(\M_L-\M_R)^p(\partial_{t^*_{m,l}}(\M_L-\M_R))(\M_L-\M_R)^{n-p-1}\L^k
+(\M_L-\M_R)^n(\partial_{t^*_{m,l}}\L^k)\Big]_-,\L\right\}\\&&-\{((\M_L-\M_R)^n\L^k)_{-},\partial_{t^*_{m,l}} \L\}\\
&&+\left\{\Big[\sum_{p=0}^{m-1}
(\M_L-\M_R)^p(\partial_{t^*_{n,k}}(\M_L-\M_R))(\M_L-\M_R)^{m-p-1}\L^l
+(\M_L-\M_R)^m(\partial_{t^*_{n,k}}\L^l)\Big]_{-},\L\right\}\\&&+
\{((\M_L-\M_R)^m\L^l)_{-},\partial_{t^*_{n,k}} \L\}\\
&=&\{[(nl-km)(\M_L-\M_R)^{m+n-1}\L^{k+l-1}]_-,\L\}\\
&=&(km-nl)\d^*_{m+n-1,k+l-1}\L.
\end{eqnarray*}
\end{proof}

From Proposition \ref{WinfiniteCalgebra}, it can be seen that the additional symmetry has a nice Block type Lie algebraic structure
whose structure theory and  representation theory have recently received much attention.
The difference of this Block Lie algebra from the one in \cite{ourBlock} is the representation space here is functional space and the one in  \cite{ourBlock} is space of operators. Similarly as \cite{ourBlock}, the action of this kind of Lie algebra on tau function space is still hard to handle with.

\sectionnew{Quasi-classical limit of the BTH}
To  consider the quasi-classical limit of the BTH, it is convenient to introduce the  order and the principal symbol of functions of difference operators.
The  order and the principal symbol are
defined for the difference operators as follows\cite{Takasaki}.

Define the order of operator $a_{n,m}(t,x) \ep^n e^{m \epsilon \d_x}$ as following
$$
    ord \left(
    \sum a_{n,m}(t,x) \ep^n e^{m \epsilon \d_x} \right)
    =\max \{ n \,|\, a_{n,m}(t,x) \neq 0\}.
$$ The {\it
principal symbol} of a difference operator $A = \sum
a_{n,m}\epsilon^n \exp(m\epsilon\d_x )$ is defined as
$$
  \sigma ^{\epsilon}(A)
   = \epsilon^{-ord(A)} \sum_{n = ord(A)} \sum_m a_{n,m}  k^m.
$$

To see the limit, we need to recall the Lax operator of the BTH given by the Laurent polynomial of $\Lambda$ \cite{C}
\begin{equation}\label{LBTH}
L:=\Lambda^{N}+u_{N-1}\Lambda^{N-1}+\cdots+u_0+\cdots + u_{-M}
\Lambda^{-M}.
\end{equation}
 The $L$ can be
written in two different ways by dressing the shift operator
\begin{equation}\label{two dressing}
L=\P_L\Lambda^N\P_L^{-1} = \P_R \Lambda^{-M}\P_R^{-1},
\end{equation}
where the dressing operators have the form,
 \begin{align}
 \P_L&=1+w_1\Lambda^{-1}+w_2\Lambda^{-2}+\ldots,
\label{dressP}\\[0.5ex]
 \P_R&=\tilde{w_0}+\tilde{w_1}\Lambda+\tilde{w_2}\Lambda^2+ \ldots.
\label{dressQ}
\end{align}
Eq.\eqref{two dressing} are quite important because it gives the reduction condition from the two-dimensional Toda lattice hierarchy.
 The
pair is unique up to multiplying $\P_L$ and $\P_R$ from the right
 by operators in the form  $1+
a_1\Lambda^{-1}+a_2\Lambda^{-2}+...$ and $\tilde{a}_0 +
\tilde{a}_1\Lambda +\tilde{a}_2\Lambda^2+\ldots$ respectively with
coefficients independent of $x$. Given any difference operator $A=
\sum_k A_k \Lambda^k$, the positive and negative projections are
defined by $A_+ = \sum_{k\geq0} A_k \Lambda^k$ and $A_- = \sum_{k<0}
A_k \Lambda^k$.

To write out  explicitly the Lax equations  of BTH,
  fractional powers $\L^{\frac1N}$ and
$L^{\frac1M}$ are defined by
\begin{equation}
  \notag
  L^{\frac1N} = \Lambda+ \sum_{k\leq 0} a_k \Lambda^k , \qquad L^{\frac1M} = \sum_{k \geq -1} b_k
  \Lambda^k,
\end{equation}
with the relations
\begin{equation}
  \notag
  (L^{\frac1N} )^N = (\L^{\frac1M} )^M = \L.
\end{equation}
Acting on free function, these two fraction powers can be seen as two different locally expansions around zero and infinity respectively.
It was  stressed that $\L^{\frac1N}$ and $\L^{\frac1M}$ are two
different operators even if $N=M(N, M\geq 2)$ in \cite{C} due to two different dressing operators. They can also be
expressed as following
\begin{equation}
\notag
  L^{\frac1N} = \P _{L}\Lambda\P_{L}^{-1}, \qquad L^{\frac1M} = \P_{R}\Lambda^{-1} \P_{R}^{ -1}.
\end{equation}
\begin{definition} \label{deflax2}
The  bigraded Toda hierarchy consists of the system
of flows given in the Lax pair formalism by
\begin{equation}
  \label{edef'}
\epsilon\frac{\partial L}{\partial t_{\alpha, n}} = [ A_{\alpha,n}
,L ]
\end{equation}
for $\alpha = N,N-1,N-2, \dots , -M+1$ and $n \geq 0$. The  operators $A_{\alpha ,n}$ are defined by
\begin{align}
  &A_{\alpha,n} = ( L^{n+1-\frac{\alpha-1}N })_+ \quad \text{for} \quad \alpha = N,N-1, \dots, 1\\
  &A_{\alpha,n} = -( L^{n+1+\frac{\alpha}M })_- \quad \text{for} \quad \alpha = 0, \dots,
  -M+1.
\end{align}
\end{definition}

The coefficients $u_i$ of $L$ are set to be regular to $\epsilon$,
i.e. $u_i(\epsilon,t,x)=u_i^0(t,x)+O(\epsilon)$.
\begin{proposition} \label{propZS2}
If $L$ satisfies the Lax equations then we have the following
Zakharov-Shabat equations
\begin{equation}
\label{zs2}  \epsilon(A_{\alpha,m})_{t_{\beta,n}} - \epsilon(
A_{\beta, n})_{t_{\alpha,m}} + [ A_{\alpha,m} , A_{\beta,n} ] =0
\end{equation}
for $-M+1 \leq \alpha, \beta \leq N$ ,  $m, n \geq 0$.
\end{proposition}

 Using the Zakharov-Shabat eqs.(\ref{zs}) we can prove that the
flows of eqs.\eqref{edef'} can commute pairwise.
\begin{lemma} \label{zs corollary2}
\begin{align}\label{compatible zero curvature2}
 &\epsilon\partial_{\beta,n}(B_{\alpha,m})_{-} - \epsilon\partial_{\alpha,m}(B_{\beta,
n})_{-} - [ (B_{\alpha,m})_{-} , (B_{\beta,n})_{-} ] =
0\\
 &-\epsilon\partial_{\beta,n}(B_{\alpha,m})_{+} +\epsilon \partial_{\alpha,m}(B_{\beta,
n})_{+} - [ (B_{\alpha,m})_{+} , (B_{\beta,n})_{+} ] =0
\end{align}
here, $-M+1\leq\alpha,\beta\leq N$,  $m,n \geq 0$.
\end{lemma}
where
\begin{equation}
  B_{\gamma , n} :=
\begin{cases}
  L^{n+1-\frac{\gamma-1}{N}} &\gamma=N\dots 1\\
L^{n+1+\frac{\gamma}{M}} &\gamma = 0\dots -M+1.
  \end{cases}
\end{equation}
\begin{theorem}
\label{t1} $L$ is a solution to the BTH if and only if there is a
pair of dressing operators $\P_L$ and $\P_R$, which satisfies the
following Sato  equations: {\allowdisplaybreaks}
\begin{eqnarray}
\label{bn12}
\epsilon\d_{\gamma,n}\P_L & =- ( B_{\gamma,n}  )_- \P_L, \\
\label{bn1'2} \epsilon\d_{\gamma,n}\P_R & = ( B_{\gamma ,n})_+\P_R,
\end{eqnarray}
where, $-M+1\leq\gamma\leq N$,  $n \geq 0$.
\end{theorem}

In paper\cite{ourBlock}, we defined the Orlov-Schulman's $M_L$, $M_R$
 operators as following
\begin{eqnarray}\label{Moperator}
&&M_L=\P_L\Gamma_L \P_L^{-1}, \ \ \ \ \ \ \ M_R=\P_R\Gamma_R
\P_R^{-1},
\end{eqnarray}
where
\begin{eqnarray}
 && \Gamma_L=
\frac{x}{N\epsilon}\Lambda^{-N}+\sum_{n\geq 0}\sum_
{\alpha=1}^{N}(n+1 - \frac{\alpha-1}{N} )
\ep^{-1}\Lambda^{N({n-\frac{\alpha-1}{N}})}t_{\alpha, n},
%\end{eqnarray}
%and
%\begin{eqnarray}\label{Moperator}
%&&M_R=\P_R\Gamma_R \P_R^{-1},
%\end{eqnarray}
%where
\\%\begin{eqnarray}
&& \Gamma_R=-\frac{x}{M\epsilon}\Lambda^M-\sum_{n\geq 0}\sum_
{\beta=-M+1}^{0}(n+1 + \frac{\beta}{M} )
\ep^{-1}\Lambda^{-M({n+\frac{\beta}{M}})}t_{\beta, n}.
\end{eqnarray}
Therefore $M_L$ and $M_R$ can be written in another form as
following
\begin{eqnarray}\label{Moperator2}
&&M_L= \frac{x}{N}L_L^{-1}+\sum_ {m=0}^{\infty}\sum_
{\gamma=1}^{N}v_{\gamma,m}(t,x)L^{-(m+2+\frac {1-\gamma}
{N})}+\sum_{n\geq 0}\sum_ {\alpha=1}^{N}(n+1-\frac{\alpha-1}{N})
L^{n-\frac{\alpha-1}{N}}t_{\alpha, n},
\end{eqnarray}
\begin{eqnarray}\label{Moperator2'}
&&M_R=-\frac{x}{M}L_R^{-1}+\sum_ {m=0}^{\infty}\sum_
{\gamma=-m+1}^{0}\bar v_{\gamma,m}(t,x)L^{-(m+2+\frac {\gamma}
{M})}-\sum_{n\geq 0}\sum_
{\beta=-M+1}^{0}({n+1+\frac{\beta}{M}})L^{n+\frac{\beta}{M}}t_{\beta,
n},
\end{eqnarray}
where $L_L^{-1}:=\P_L\La^{-N} \P_L^{-1}$, $L_R^{-1}:=\P_R\La^{M}
\P_R^{-1}$.

To consider the quasi-classical limit, we set
\begin{eqnarray}
\P_L&=&\exp(\epsilon^{-1}X_L(\epsilon,t,x)),\\
\P_R&=&\exp(\phi(\epsilon,t,x))\exp(\epsilon^{-1}X_R(\epsilon,t,x)),\\
ord(X_L(\epsilon,t,x))&=&ord(X_R(\epsilon,t,x))=0,
\end{eqnarray}
\begin{eqnarray}
\sigma^{\epsilon}(X_L)=\varphi_L,\ \
\sigma^{\epsilon}(X_R)=\varphi_R,
\end{eqnarray}
where   $ ord(\phi(\epsilon,t,x))\leq 0$  and
\begin{eqnarray}\label{dis phi flow 0}
    \d_{t_{\alpha,n}}  \phi& =& (\B_{\alpha,n} )_0,\\
\label{dis phi x flow}
   \frac1\epsilon( \phi(\epsilon,t,x)-\phi(\epsilon,t,x-\epsilon))& = &\frac1M\log u_{-M},
\end{eqnarray}
where $(\ \ \ )_0$ is projection to operators not containing
$\Lambda$, i.e. $\Lambda^0$ term.

Now, we  do the following change $[\ ,\ ]\rightarrow\{\ ,\ \},
\epsilon^{-1}X_L(\epsilon,t,x)\rightarrow\varphi_L,
\epsilon^{-1}X_R(\epsilon,t,x)\rightarrow\varphi_R,
\phi(\epsilon,t,x)\rightarrow\phi(t,x), \Lambda\rightarrow k$.
Then we have the following two propositions with quasi-classical limit.
\begin{proposition}
Symbols $\sigma^{\epsilon}(X_L), \sigma^{\epsilon}(X_R),
\sigma^{\epsilon}(\phi)$ give dressing functions $\varphi_L,
\varphi_R$ and potential function $\phi$ of the dispersionless Lax
function $\L=\sigma^{\epsilon}(L)$ respectively. Conversely if
$\varphi_L, \varphi_R$ and $\phi_0$ are  dressing functions and
potential function of dispersionless BTH respectively, then there
exist a solution $L$ of the dBTH and dressing operators
$\P_L=\exp(\epsilon^{-1}X_L(\epsilon,t,x))$ and
$\P_R=\exp(\phi)\exp(\epsilon^{-1}X_R(\epsilon,t,x))$ such that
$\sigma^{\epsilon}(L)=\L, \sigma^{\epsilon}(X_L)=\varphi_L,
\sigma^{\epsilon}(X_R)=\varphi_R, \sigma^{\epsilon}(\phi)=\phi_0$.
\end{proposition}

\begin{proposition}
$ord^{\epsilon}(\epsilon M_L)=ord^{\epsilon}(\epsilon M_R)=0 $ and
$\M_L=\sigma^{\epsilon}(\epsilon M_L)$ and
$\M_R=\sigma^{\epsilon}(\epsilon M_R)$ are the Orlov functions of
the dispersionless BTH whose Lax function is
$\L=\sigma^{\epsilon}(L)$.
\end{proposition}

To see the limit clearly,
we firstly introduce spectral $z$ and two
functions
 $w_L(t,x,z)$ and $w_R(t,x,z)$ which  have forms
\begin{eqnarray}
w_L(t,x,z) &=&\P_L(x,\Lambda)e^{\xi_L(t,x,z)},\\
w_{R}(t,x,z) &=&\P_R(x,\Lambda)e^{\xi_R(t,x,z)},
\end{eqnarray}
where
\begin{eqnarray}
\xi_L(t,x,z) &=&\sum_{n\geq 0}\sum_
{\alpha=1}^{N}\ep^{-1}z^{n+1-\frac{\alpha-1}{N}}t_{\alpha, n}
+\frac{x}{N\epsilon}\log z,\\
\xi_{R}(t,x,z) &=&-\sum_{n\geq 0}\sum_ {\beta=-M+1}^{0}
\ep^{-1}z^{n+1+\frac{\beta}{M}}t_{\beta, n}-\frac{x}{M\epsilon}\log
z.
\end{eqnarray}
We call these two functions $w_L(t,x,z)$ and $w_{R}(t,x,z)$ wave functions. These two wave functions are a little different from ones in \cite{ourBlock}.
Similarly as \cite{ourBlock}, the following proposition holds.
\begin{proposition}\label{linearfunction}
The functions $w_L(t,x,z)$ and $w_R(t,x,z)$  satisfy the following linear equations
\begin{eqnarray}
\begin{cases}\label{wLlinear}
Lw_L(t,x,z)=&zw_L(t,x,z),\\
M_Lw_L(t,x,z)=&\partial_zw_L(t,x,z),\\
\ep \partial_{ t_{\gamma,n}}w_L(t,x,z)=&A_{\gamma,n}w_L(t,x,z),
 \end{cases}
 \end{eqnarray}
 \begin{eqnarray}
 \begin{cases}\label{wRlinear}
 Lw_{R}(t,x,z)=&zw_{R}(t,x,z),\\
M_{R}w_{R}(t,x,z)=&\partial_{z}w_{R}(t,x,z),\\
\ep \partial_{ t_{\gamma,n}}w_{R}(t,x,z)=&A_{\gamma,n}w_{R}(t,x,z),
 \end{cases}
\end{eqnarray}
where, $ -M+1  \leq \gamma \leq N, n\geq 0$.
\end{proposition}
\begin{proof}
The proof is similar as paper \cite{ourBlock}.
\end{proof}
Using Proposition \ref{linearfunction}, $w_L(\epsilon,t,z)$ and $w_R(\epsilon,t,z)$ can be written using  functions $S_L(\epsilon,t,x,z)$ and $S_R(\epsilon,t,x,z)$ in the following proposition.
\begin{proposition}
The Baker function $w_L(\epsilon,t,x,z)$ and $w_R(\epsilon,t,x,z)$ take
a WKB asymptotic form as $\epsilon \rightarrow 0$:
\begin{eqnarray}\label{ WKB asymptotic}
w_L(\epsilon,t,x,z)&=&\exp\left(\epsilon^{-1}S_L(t,x,z)+O(\epsilon^0)\right),\\\label{
WKB asymptotic'}
w_R(\epsilon,t,x,z)&=&\exp\left(\epsilon^{-1}S_R(t,x,z)+O(\epsilon^0)\right),
\end{eqnarray}
where\begin{eqnarray}\label{SL}
 S_L(t,x,z)&=&\sum_{n\geq 0}\sum_
{\alpha=1}^{N}z^{({n+1-\frac{\alpha-1}{N}})}t_{\alpha, n}
+\frac{x}{N}\log z -\sum_ {m=0}^{\infty}\sum_
{\gamma=1}^{N}v_{\gamma,m}(t,x)\frac{ z^{-(m+1+\frac
{1-\gamma} {N})}}{m+1+\frac {1-\gamma} {N}},
\end{eqnarray}
\begin{eqnarray}\label{SR}&&\\ \notag
S_R(t,x,z)&=&-\sum_{n\geq 0}\sum_ {\beta=-M+1}^{0}
z^{n+1+\frac{\beta}{M}}t_{\beta, n}-\frac{x}{M}\log z +\sum_
{m=0}^{\infty}\sum_
{\gamma=-M+1}^{0}v_{\gamma,m}(t,x)\frac{z^{-(m+1+\frac {\gamma}
{M})}}{m+1+\frac {\gamma} {M}}.
\end{eqnarray}
\end{proposition}
Then we get following proposition
similar as \cite{Takasaki}
\begin{proposition}
The spectral $z$ of Lax function has following representations
\begin{eqnarray*}  z
&=&e^{N\d_xS_L(x,t,z)} +u_{N-1}e^{(N-1)\d_xS_L(x,t,z)} +\dots+ u_{-M}
e^{-M\d_xS_L(x,t,z)}=\sigma^{\epsilon}(L)|_{k=\d_xS_L(x,t,z)},
\end{eqnarray*}

\begin{eqnarray*}z
=e^{N\d_xS_R(x,t,z)} +u_{N-1}e^{(N-1)\d_xS_R(x,t,z)} +\dots+ u_{-M}
e^{-M\d_xS_R(x,t,z)}=\sigma^{\epsilon}(L)|_{k=\d_xS_R(x,t,z)},
\end{eqnarray*}
and derivatives of S function have formulas
\begin{eqnarray}
d S_L(x,t,z) =\M_L(z)dz +\frac{\d S_L}{\d x}d x+\sum_{n\geq 0}\sum_
{\gamma=-M+1}^{N}\A_{\gamma,n}(e^{\d_x S_L})dt_{\gamma,n},
\end{eqnarray}
\begin{eqnarray}
d S_R(x,t,z) =\M_R(z)dz^{-1} +\frac{\d S_R}{\d x}d x+\sum_{n\geq
0}\sum_ {\gamma=-M+1}^{N}\A_{\gamma,n}(e^{\d_x S_R})dt_{\gamma,n}.
\end{eqnarray}
\end{proposition}
\begin{proof}
By eqs.\eqref{wLlinear}, we can do the following computation
\begin{eqnarray*}z w_L(t,x,z)&=&L w_L(t,x,z)=(\Lambda^{N}+u_{N-1}\Lambda^{N-1}+\dots + u_{-M}
\Lambda^{-M})\exp\left(\epsilon^{-1}S_L(x,t,z)+O(\epsilon^0)\right).
\end{eqnarray*}
After that, we can pursue the WKB analysis, regarding $S_L(t,x,z)$ as
the phase function which further leads to
\begin{eqnarray}&&\\ \notag z
&=&e^{N\d_xS_L(x,t,z)} +u_{N-1}e^{(N-1)\d_xS_L(x,t,z)} +\dots+ u_{-M}
e^{-M\d_xS_L(x,t,z)}=\sigma^{\epsilon}(L)|_{k=\d_xS_L(x,t,z)}.
\end{eqnarray}
By eqs.\eqref{wLlinear}, we can also get
\begin{eqnarray*}
&& M_Lw_L(t,x,z)\\
&=&
\partial_zw_L(t,x,z)=\partial_z\exp\left(\epsilon^{-1}S_L(x,t,z)+O(\epsilon^0)\right)\\
&=&(\epsilon^{-1}\partial_zS_L(x,t,z)+O(\epsilon^0))e^{\epsilon^{-1}S_L(x,t,z)+o(\epsilon)}\\
&=& \epsilon^{-1}(\sum_{n\geq 0}\sum_
{\alpha=1}^{N}(n+1-\frac{\alpha-1}{N})z^{({n-\frac{\alpha-1}{N}})}t_{\alpha,
n} +\frac{x}{N z} -\sum_ {m=0}^{\infty}\sum_
{\gamma=1}^{N}v_{\gamma,m}(t,x)z^{-(m+2+\frac {1-\gamma}
{N})})w_L(t,x,z)\\
&&+O(\epsilon^0)w_L(t,x,z),
\end{eqnarray*}
which implies $M_Lw_L(t,x,z)=\frac{\M_L}{\epsilon}w_L(t,x,z)
=\ep^{-1}\sigma^{\epsilon}(M_L)|_{L=z}w_L(t,x,z)$. Also by the
equations above, we can get $\partial_zS_L(x,t,z)
=\M_L=\sigma^{\epsilon}(M_L)|_{L=z}$.

 Because
\begin{eqnarray}
 \epsilon \partial_{ t_{\gamma,n}}w_L(t,x,z)&=&\epsilon\partial_{
 t_{\gamma,n}}\exp\left(\epsilon^{-1}S_L(x,t,z)+O(\epsilon^0)\right)\\ \notag&
 =&(\partial_{
 t_{\gamma,n}}S_L(x,t,z)+O(\epsilon^0))\exp\left(\epsilon^{-1}S_L(x,t,z)+O(\epsilon^0)\right)
 \end{eqnarray}
and
\begin{eqnarray*}
 \epsilon \partial_{ t_{\gamma,n}}w_L(t,x,z)&=&A_{\gamma,n}(\Lambda)w_L(t,x,z)\\
  &=&A_{\gamma,n}(\Lambda)\exp\left(\epsilon^{-1}S_L(x,t,z)+O(\epsilon^0)\right)\\
&=&(\A_{\gamma,n}(e^{\d_x
S_L})+O(\epsilon^0))\exp\left(\epsilon^{-1}S_L(x,t,z)+O(\epsilon^0)\right),
\end{eqnarray*}
  we can get
\begin{eqnarray}\label{li}
\partial_{
 t_{\gamma,n}}S_L(x,t,z)
&=&\A_{\gamma,n}(e^{\d_x S_L}).
\end{eqnarray}
This implies
\begin{eqnarray}
\partial_{
 t_{\gamma,n}}(\d_xS_L(x,t,z))
&=&\frac{\d \A_{\gamma,n}(e^{\d_x S_L})}{\d x}.
\end{eqnarray}
So we can consider $\d_xS_L(x,t,z)$ which will be denoted as $\log
k_1$ later as conserved density.

By eqs.\eqref{wRlinear}, we can do the following computation
\begin{eqnarray*}z w_R(t,x,z)&=&L w_R(t,x,z)=(\Lambda^{N}+u_{N-1}\Lambda^{N-1}+\dots + u_{-M}
\Lambda^{-M})\exp\left(\epsilon^{-1}S_R(x,t,z)+O(\epsilon^0)\right).
\end{eqnarray*}
We can also pursue the WKB analysis, also regarding $S_R(t,x,z)$ as
the phase function which leads to
\begin{eqnarray*}z
=e^{N\d_xS_R(x,t,z)} +u_{N-1}e^{(N-1)\d_xS_R(x,t,z)} +\dots+ u_{-M}
e^{-M\d_xS_R(x,t,z)}=\sigma^{\epsilon}(L)|_{k=\d_xS_R(x,t,z)}.
\end{eqnarray*}
Similar result  can also be got as following
\begin{eqnarray}
 M_Rw_R(t,x,z)&=& \frac{\M_R}{\epsilon}w_R(t,x,z)
=\ep^{-1}\sigma^{\epsilon}(M_R)|_{L=z}w_R(t,x,z).
\end{eqnarray}
 Also by the
equations above, we can get $\partial_zS_R(x,t,z)
=\M_R=\sigma^{\epsilon}(M_R)|_{L=z}$. Direct calculation can lead to
\begin{eqnarray}
 \epsilon\partial_{ t_{\gamma,n}}w_R(t,x,z)
 =&(\partial_{
 t_{\gamma,n}}S_R(x,t,z)+O(\epsilon^0))\exp\left(\epsilon^{-1}S_R(x,t,z)+O(\epsilon^0)\right).
 \end{eqnarray}

 By formula eqs.\eqref{wRlinear}, following identities hold
 \begin{eqnarray*}
 \epsilon\partial_{ t_{\gamma,n}}w_R(t,x,z)&=&A_{\gamma,n}(\Lambda)w_R(t,x,z)\\
  &=&A_{\gamma,n}(\Lambda)\exp\left(\epsilon^{-1}S_R(x,t,z)+O(\epsilon^0)\right)\\
&=&\A_{\gamma,n}(e^{\d_x
S_R})+O(\epsilon^0))\exp\left(\epsilon^{-1}S_R(x,t,z)+O(\epsilon^0)\right).
\end{eqnarray*}
 So we can get
\begin{eqnarray}\label{wang}
\partial_{
 t_{\gamma,n}}S_R(x,t,z)
&=&\A_{\gamma,n}(e^{\d_x S_R}).
\end{eqnarray}

So the following derivatives will be correct
\begin{eqnarray}
d S_L(x,t,z) =\M_L(z)dz +\frac{\d S_L}{\d x}d x+\sum_{n\geq 0}\sum_
{\gamma=-M+1}^{N}\A_{\gamma,n}(e^{\d_x S_L})dt_{\gamma,n},
\end{eqnarray}
\begin{eqnarray}
d S_R(x,t,z) =\M_R(z)dz^{-1} +\frac{\d S_R}{\d x}d x+\sum_{n\geq
0}\sum_ {\gamma=-M+1}^{N}\A_{\gamma,n}(e^{\d_x S_R})dt_{\gamma,n}.
\end{eqnarray}

\end{proof}

Then it it time to consider the Legendre-type transformation as following.
\begin{proposition}
The Legendre-type transformation $(t,x,z) \rightarrow (t,x, k_i), i=1,2$
is defined by
\begin{eqnarray}
e^{\d_x S_L(t,x,z)}=k_1,\ \ e^{\d_x S_R(t,x,z)}=k_2,
\end{eqnarray}
the spectral parameters $z$ turns into the $\L$-function of the
dispersionless BTH which has two expressions by $k_1$ and $k_2$ respectively
\begin{eqnarray}
z=\L(t,x,k_1)&=&k^{N}_1+u_{N-1}k^{N-1}_1+\dots + u_{-M} k^{-M}_1,\\
z=\L(t,x,k_2)&=&k^{N}_2+u_{N-1}k^{N-1}_2+\dots + u_{-M} k^{-M}_2,
\end{eqnarray}
and $S_L, S_R$ become the corresponding $S$-functions.
\end{proposition}

After these, we will start to consider the free energy of the dBTH in the next section.

\sectionnew{Tau function and free energy of dBTH}

As paper \cite{ourBlock}, the relation of  tau function for the BTH with
wave functions is as following
\begin{eqnarray*}
w_L(t,x,z) &=&\P_L(x,\Lambda)\exp(\sum_{n\geq 0}\sum_
{\alpha=1}^{N}\ep^{-1}z^{({n+1-\frac{\alpha-1}{N}})}t_{\alpha, n}
+\frac{x}{N\epsilon}\log z)\\
&=&\P_L(x,\lambda^{\frac1N})\exp(\frac{1}{\epsilon}(t_L(z^{\frac1N})+\frac{x}{N}\log z))\\
&=&\frac{\tau(x-\epsilon/2,t-[z^{-\frac1N}]^N)}{\tau(x-\epsilon/2,t)}
\exp(\frac{1}{\epsilon}(t_L(z^{\frac1N})+\frac{x}{N}\log z)),\\
w_{R}(t,x,z)
&=&\P_R(x,\Lambda)\exp(-\sum_{n\geq 0}\sum_
{\beta=-M+1}^{0}\ep^{-1}z^{n+1+\frac{\beta}{M}}t_{\beta,
n}-\frac{x}{M\epsilon}\log z)\\
&=&\P_R(x,z^{-\frac1M})\exp(\frac{1}{\epsilon}(-t_R(z^{\frac1M})-\frac{x}{M}\log z))\\
&=&\frac{\tau(x+\epsilon/2,t+[z^{-\frac1M}]^M)}{\tau(x-\epsilon/2,t)}
\exp(-\frac{1}{\epsilon}t_R(z^{\frac1M})-\frac{x}{M}\log z),
\end{eqnarray*}
 where
     \begin{eqnarray}
 \[z^{-\frac1N}\]^{N}_{\alpha,n} :=
\begin{cases}
  \frac{\ep z^{-(n+1-\frac{\alpha-1}{N})}}{N(n+1 - \frac{\alpha-1}{N}
)}, &\alpha=N,N-1,\dots 1,\\
 0, &\alpha = 0, -1\dots -M+1,
  \end{cases}
\end{eqnarray}
 \begin{eqnarray}
\[z^{-\frac1M}\]^{M}_{\alpha,n} :=
\begin{cases}
0, &\alpha=N, N-1,\dots 1,\\
\frac{\epsilon z^{-(n+1+\frac{\beta}{M})}}{M(n+1
+\frac{\beta}{M} )}, &\alpha = 0, -1, \dots -M+1,
  \end{cases}
\end{eqnarray}
Define free energy $F(t,x)$ as following
 \begin{eqnarray}
\log \tau(\epsilon,t,x)=\epsilon^{-2}F(t,x)+0(\epsilon^{-1}) \ \ \
(\epsilon\rightarrow 0).
\end{eqnarray}
 Taking the logarithm of identities above and comparing them with
the form eq.\eqref{ WKB asymptotic}, one can find that $\log \tau$
should behave as

\begin{eqnarray*}\notag &&\epsilon^{-1}(\sum_{n\geq 0}\sum_
{\alpha=1}^{N}z^{({n+1-\frac{\alpha-1}{N}})}t_{\alpha, n}
+\frac{x}{N}\log z -\sum_ {m=0}^{\infty}\sum_
{\gamma=1}^{N}v_{\gamma,m}(t,x)\frac{ z^{-(m+1+\frac {1-\gamma}
{N})}}{m+1+\frac {1-\gamma} {N}})+O(\epsilon^0)\\&=& \sum_{n\geq
0}\sum_
{\alpha=1}^{N}\frac{z^{({n+1-\frac{\alpha-1}{N}})}}{\epsilon}t_{\alpha,
n} +\frac{x}{N\epsilon}\log z+\log
\tau(\epsilon,t-[z^{-\frac1N}]^N,x)-\log \tau(\epsilon,t,x),
\end{eqnarray*}
which further leads to following relation
\begin{eqnarray}\label{valpham1}
 v_{\alpha,m}(t,x) =\d_{t_{\alpha,m}}F,\ \  1\leq \alpha\leq N, m\geq 0.
\end{eqnarray}
Similarly we can also get
\begin{eqnarray}\label{barvalpham}
\bar v_{\beta,m}(t,x) =\d_{t_{\beta,m}}F,\ \  -M+1\leq \beta\leq 0, m\geq 0.
\end{eqnarray}
Considering Lemma \ref{vderivative}, the second derivatives of free
energy will have  formula in following lemma.
\begin{lemma}
The derivatives of free energy have the following formula
\begin{eqnarray}
F_{\alpha,n;\beta,m} =Res_{k}\B_{\alpha,n} d_k (\B_{\beta,m})_+,\ \  -M+1\leq \alpha,\beta\leq N, m,n\geq 0,
\end{eqnarray}
where $F_{\alpha,n;\beta,m}:=\d_{t_{\alpha,n}}\d_{t_{\beta,m}}F$.
\end{lemma}
\begin{proof}
Bringing  eq.\eqref{valpham1} and eq.\eqref{barvalpham} into  Lemma \ref{vderivative} will lead to this lemma.
\end{proof}

Eqs.\eqref{SL}, \eqref{SR} and eqs.\eqref{li}, \eqref{wang}  can lead to following proposition.
\begin{proposition}
The following identities hold
\begin{eqnarray} \A_{\gamma,n}(t,x,k_1)&=&z^{({n+1-\frac{\gamma-1}{N}})}
 -\sum_ {m=0}^{\infty}\sum_
{\alpha=1}^{N}F_{\alpha,m;\gamma,n}(t,x)\frac{ z^{-(m+1+\frac
{1-\alpha} {N})}}{m+1+\frac {1-\alpha} {N}},
\end{eqnarray}
\begin{eqnarray}\A_{\gamma,n}(t,x,k_1)&=& -\sum_ {m=0}^{\infty}\sum_
{\alpha=1}^{N}F_{\alpha,m;\gamma,n}(t,x)\frac{ z^{-(m+1+\frac
{1-\alpha} {N})}}{m+1+\frac {1-\alpha} {N}},
\end{eqnarray}
\begin{eqnarray}\notag
\A_{\gamma,n}(t,x,k_2)&=&\sum_ {m=0}^{\infty}\sum_
{\beta=-M+1}^{0}F_{\beta,m;\gamma,n}(t,x)\frac{z^{-(m+1+\frac
{\beta} {M})}}{m+1+\frac {\beta} {M}},
\end{eqnarray}
\begin{eqnarray}\notag
\A_{\gamma,n}(t,x,k_2)&=&- z^{n+1+\frac{\gamma}{M}}+\sum_
{m=0}^{\infty}\sum_
{\beta=-M+1}^{0}F_{\beta,m;\gamma,n}(t,x)\frac{z^{-(m+1+\frac
{\beta} {M})}}{m+1+\frac {\beta} {M}},
\end{eqnarray}
where $-M+1\leq \gamma \leq N, n\in \Z_+,$
\begin{eqnarray}
e^{\d_x S_L(t,x,z)}=k_1,\ \ e^{\d_x S_R(t,x,z)}=k_2,\ \ z=\L(t,x,k_1)=\L(t,x,k_2).
\end{eqnarray}
\end{proposition}
After the definition of free energy of the dBTH, the dispersionless Hirota  bilinear identities about tau
 function will be derived in the next section.

\sectionnew{Dispersionless Hirota  bilinear identities}
In our paper \cite{ourJMP}, we have got the following Hirota
bilinear identity of the BTH in one corollary
\begin{eqnarray}\notag
&& \res_{\lambda }
 \left\{ \lambda^{m-1}
\tau(x,t-[\lambda^{-1}]^N)\times
\tau(x-(m-1)\epsilon,t'+[\lambda^{-1}]^N) e^{\xi_L(t-t')} \right\}
\\\label{HBE511} &&=  \res_{\lambda }
 \left\{ \lambda^{m-1}
\tau(x+\epsilon,t+[\lambda]^M)\times
\tau(x-m\epsilon,t'-[\lambda]^M) e^{\xi_R(t-t')}\right\} ,
\end{eqnarray}
     where
     \begin{eqnarray}
 \[\lambda^{-1}\]^{N}_{\alpha,n} :=
\begin{cases}
  \frac{\ep \lambda^{-N(n+1-\frac{\alpha-1}{N})}}{N(n+1 - \frac{\alpha-1}{N}
)}, &\alpha=N,N-1,\dots 1,\\
 0, &\alpha = 0, -1\dots -M+1,
  \end{cases}
\end{eqnarray}
 \begin{eqnarray}
\[\lambda\]^{M}_{\alpha,n} :=
\begin{cases}
0, &\alpha=N, N-1,\dots 1,\\
\frac{\epsilon \lambda^{M(n+1+\frac{\beta}{M})}}{M(n+1
+\frac{\beta}{M} )}, &\alpha = 0, -1, \dots -M+1,
  \end{cases}
\end{eqnarray}

\begin{eqnarray*} \xi_L(t-t')&=&\sum_{n\geq 0} \sum_
{\alpha=1}^{N}\lambda^{N({n+1-\frac{\alpha-1}{N}})}(t_{\alpha,n}-t'_{\alpha,n}),\\
\xi_R(t-t')&=&-\sum_{n\geq 0} \sum_ {\beta=-M+1}^{0}\lambda^{-M({n+1+\frac{\beta}{M}})}(t_{\beta,n}-t'_{\beta,n}).
\end{eqnarray*}
In order to understand the properties of the dBTH more, it is useful to get the fay-like identities from the ones of BTH.

We can choose different values for $m, t, t'$ which lead to the
following dispersionless Fay-like
identities:\\

{\bf  \Rmnum{1}.} $m=0, t'=t-[\lambda_1^{-1}]^N-[\lambda_2^{-1}]^N$.
In this case the Hirota bilinear identity \eqref{HBE511} will lead to
\begin{eqnarray*}\label{HBE6}
&& \res_{\lambda }
 \left\{ \tau(x,t-[\lambda^{-1}]^N)\times
\tau(x+\epsilon,t'+[\lambda^{-1}]^N)\frac{1}{(1-\lambda
\lambda_1^{-1})(1-\lambda \lambda_2^{-1})}\frac{1}{\lambda}
\right\}  \\
&&=  \res_{\lambda }
 \left\{
\tau(x+\epsilon,t+[\lambda]^M)\times
\tau(x,t'-[\lambda]^M)\frac{1}{\lambda} \right\}.
\end{eqnarray*}
Using \begin{eqnarray*}
(1-\lambda_1^{-1}\lambda)^{-1}(1-\lambda_2^{-1}\lambda)^{-1} =
\frac{\lambda_1^{-1}}{\lambda_1^{-1}-\lambda_2^{-1}}(1-\lambda_1^{-1}\lambda)^{-1}-\frac{\lambda_2^{-1}}{\lambda_1^{-1}-\lambda_2^{-1}}
(1-\lambda_2^{-1}\lambda)^{-1}, \end{eqnarray*} we  get
\begin{eqnarray*}\label{HBE6}
&&  \frac{\lambda_1^{-1}}{\lambda_1^{-1}-\lambda_2^{-1}}
\tau(x,t-[\lambda_1^{-1}]^N) \tau(x+\epsilon,t'+[\lambda_1^{-1}]^N)
-\frac{\lambda_2^{-1}}{\lambda_1^{-1}-\lambda_2^{-1}}
\tau(x,t-[\lambda_2^{-1}]^N)
\tau(x+\epsilon,t'+[\lambda_2^{-1}]^N)  \\
&&= \tau(x+\epsilon,t) \tau(x,t').
\end{eqnarray*}
It further leads to
\begin{eqnarray}\notag
\frac{\lambda_1^{-1}}{\lambda_1^{-1}-\lambda_2^{-1}}
\tau(x,t-[\lambda_1^{-1}]^N) \tau(x+\epsilon,t-[\lambda_2^{-1}]^N)
-\frac{\lambda_2^{-1}}{\lambda_1^{-1}-\lambda_2^{-1}}
\tau(x,t-[\lambda_2^{-1}]^N) \tau(x+\epsilon,t-[\lambda_1^{-1}]^N)
\\ \label{HBE61} = \tau(x+\epsilon,t)
\tau(x,t-[\lambda_1^{-1}]^N-[\lambda_2^{-1}]^N).
\end{eqnarray}
Dividing both sides of eq.\eqref{HBE61}  by $\tau(x,t-[\lambda_1^{-1}]^N) \tau(x+\epsilon,t-[\lambda_2^{-1}]^N)$will lead to the following
equation
\begin{eqnarray}\label{HBE610}\notag
&&  \frac{\lambda_1^{-1}}{\lambda_1^{-1}-\lambda_2^{-1}}
-\frac{\lambda_2^{-1}}{\lambda_1^{-1}-\lambda_2^{-1}}
\frac{\tau(x,t-[\lambda_2^{-1}]^N)
\tau(x+\epsilon,t-[\lambda_1^{-1}]^N)}{\tau(x,t-[\lambda_1^{-1}]^N) \tau(x+\epsilon,t-[\lambda_2^{-1}]^N)}  \\
&&= \frac{\tau(x+\epsilon,t)
\tau(x,t-[\lambda_1^{-1}]^N-[\lambda_2^{-1}]^N)}{\tau(x,t-[\lambda_1^{-1}]^N)
\tau(x+\epsilon,t-[\lambda_2^{-1}]^N)}.
\end{eqnarray}
Finally we can get the following bilinear equations on free energy,
\begin{eqnarray*}\label{HBE6107}\notag
&&  \lambda_1^{-1}\exp\(-\sum_{n\geq 0}\sum_
{\alpha=1}^{N}\frac{\lambda_2^{-N(n+1-\frac{{\alpha-1}}{N})}
}{N(n+1- \frac{\alpha-1}{N} )}
 \frac{\partial^2
F}{\partial{ t_{\alpha ,n}}\partial{x}}\) -\lambda_2^{-1}
\exp\(-\sum_{n\geq 0}\sum_
{\alpha=1}^{N}\frac{\lambda_1^{-N(n+1-\frac{{\alpha-1}}{N})}}{N(n+1-
\frac{\alpha-1}{N} )}
 \frac{\partial^2 F}{\partial{
t_{\alpha ,n}}\partial{x}}\) \\
&&= (\lambda_1^{-1}-\lambda_2^{-1})\exp\(\sum_{m,n\geq 0}\sum_
{\alpha,\alpha'=1}^{N}\frac{\lambda_1^{-N(n+1-\frac{{\alpha-1}}{N})}\lambda_2^{-N(m+1-\frac{{\alpha'-1}}{N})}
}{N^2(n+1 - \frac{\alpha-1}{N} )(m+1 - \frac{\alpha'-1}{N} )}
 \frac{\partial^2 F}{\partial{
t_{\alpha ,n}}\partial{
t_{\alpha' ,m}}}\).
\end{eqnarray*}

{\bf  \Rmnum{2}.} $m=0, t'=t+[\lambda_1]^M+[\lambda_2]^M$. In this
case the Hirota bilinear identity \eqref{HBE511} will lead to
\begin{eqnarray*}\label{HBE62}
&& \res_{\lambda }
 \left\{ \tau(x,t-[\lambda^{-1}]^N)\times
\tau(x+\epsilon,t'+[\lambda^{-1}]^N)\lambda^{-1}
\right\} \\
&&=  \res_{\lambda }
 \left\{
\tau(x+\epsilon,t+[\lambda]^M)\times
\tau(x,t'-[\lambda]^M)\lambda^{-1}\frac{1}{(1-\lambda^{-1}
\lambda_1)(1-\lambda^{-1} \lambda_2)} \right\}.
\end{eqnarray*}
Using \begin{eqnarray*}
(1-\lambda^{-1}\lambda_1)^{-1}(1-\lambda_2\lambda^{-1})^{-1} =
\frac{\lambda_1}{\lambda_1-\lambda_2}(1-\lambda_1\lambda^{-1})^{-1}-\frac{\lambda_2}{\lambda_1-\lambda_2}
(1-\lambda_2\lambda^{-1})^{-1}, \end{eqnarray*} we  get
\begin{eqnarray*}\label{HBE62}
&&\tau(x,t) \tau(x+\epsilon,t')\\
&=& \frac{\lambda_1}{\lambda_1-\lambda_2}
\tau(x+\epsilon,t+[\lambda_1]^M) \tau(x,t'-[\lambda_1]^M)
-\frac{\lambda_2}{\lambda_1-\lambda_2}
\tau(x+\epsilon,t+[\lambda_2]^M) \tau(x,t'-[\lambda_2]^M).
\end{eqnarray*}
It further leads to
\begin{eqnarray}\label{HBE62'}
&&\tau(x,t) \tau(x+\epsilon,t+[\lambda_1]^M+[\lambda_2]^M)\\\notag
&=& \frac{\lambda_1}{\lambda_1-\lambda_2}
\tau(x+\epsilon,t+[\lambda_1]^M) \tau(x,t+[\lambda_2]^M)
-\frac{\lambda_2}{\lambda_1-\lambda_2}
\tau(x+\epsilon,t+[\lambda_2]^M) \tau(x,t+[\lambda_1]^M).
\end{eqnarray}

Finally we can get the following bilinear equations on free energy,
\begin{eqnarray}\label{HBE62'1}
&&\exp\( \sum_{m,n\geq 0}\sum_
{\beta,\beta'=-M+1}^{0}\frac{\lambda_1^{M(n+1+\frac{{\beta}}{M})}\lambda_2^{M(m+1+\frac{{\beta'}}{M})}
}{M^2(n+1 + \frac{\beta}{M} )(m+1 + \frac{\beta'}{M} )}
 \frac{\partial^2 F}{\partial{
t_{\beta ,n}}\partial{ t_{\beta' ,m}}}\)\\\notag &=&
\frac{\lambda_1}{\lambda_1-\lambda_2}\exp\(-\sum_{n\geq 0}\sum_
{\beta=-M+1}^{0}\frac{\lambda_2^{M(n+1+\frac{{\beta}}{M})} }{M(n+1 +
\frac{\beta}{M} )}
 \frac{\partial^2 F}{\partial{
t_{\beta ,n}}\partial{x}}\) \\ \notag
&&-\frac{\lambda_2}{\lambda_1-\lambda_2} \exp\(-\sum_{n\geq 0}\sum_
{\beta=-M+1}^{0}\frac{\lambda_1^{M(n+1+\frac{{\beta}}{M})}}{M(n+1 +
\frac{\beta}{M} )}
 \frac{\partial^2 F}{\partial{
t_{\beta ,n}}\partial{x}}\).
\end{eqnarray}

{\bf  \Rmnum{3}.} $m=1, t'=t-[\lambda_1^{-1}]^N+[\lambda_2]^M$. In
this case the Hirota bilinear identity \eqref{HBE511} will lead to
\begin{eqnarray*}\label{HBE62}
&& \res_{\lambda }
 \left\{ \tau(x,t-[\lambda^{-1}]^N)\times
\tau(x,t'+[\lambda^{-1}]^N)\frac{1}{1-\lambda \lambda_1^{-1}}
\right\} \\
&&=  \res_{\lambda }
 \left\{
\tau(x+\epsilon,t+[\lambda]^M)\times
\tau(x-\epsilon,t'-[\lambda]^M)\frac{1}{1-\lambda^{-1} \lambda_2}
\right\} ,
\end{eqnarray*}
which is equivalent to
\begin{eqnarray*}\label{HBE62'''}
 \lambda_1 (\tau(x,t-[\lambda_1^{-1}]^N) \tau(x,t'+[\lambda_1^{-1}]^N)-\tau(x,t) \tau(x,t')) =
\lambda_2\tau(x+\epsilon,t+[\lambda_2]^M)
\tau(x-\epsilon,t'-[\lambda_2]^M).
\end{eqnarray*}
It further implies
\begin{eqnarray}\notag
 \lambda_1 (\tau(x,t-[\lambda_1^{-1}]^N)\tau(x,t+[\lambda_2]^M)-\tau(x,t) \tau(x,t-[\lambda_1^{-1}]^N+[\lambda_2]^M))
 \\\label{HBE62''}
=\lambda_2\tau(x+\epsilon,t+[\lambda_2]^M)
\tau(x-\epsilon,t-[\lambda_1^{-1}]^N).
\end{eqnarray}

Similarly we can also get the following Hirota bilinear equation on
free energy,
\begin{eqnarray}\notag
 &&1-\exp(\sum_{m,n\geq 0}\sum_
{\beta=-M+1}^{0}\sum_
{\alpha=1}^{N}\frac{\lambda_1^{-N(n+1-\frac{{\alpha-1}}{N})}
\lambda_2^{M(m+1+\frac{{\beta}}{M})} }{NM(n+1- \frac{\alpha-1}{N}
)(n+1 + \frac{\beta}{M} )}
 \frac{\partial^2 F}{\partial{
t_{\alpha,n}}\partial{
t_{\beta ,m}}})\\ \notag
&=& \lambda_1^{-1}\lambda_2\exp\left(\frac{\partial^2
F}{\partial^2{x}}+\sum_{n\geq 0}\sum_ {\alpha=1}^{N}\frac{\lambda_1^{-N(n+1-\frac{{\alpha-1}}{N})}}{N(n+1- \frac{\alpha-1}{N} )}
 \frac{\partial^2 F}{\partial{
t_{\alpha ,n}}\partial{x}}\right.\\
&&\left.+\sum_{n\geq 0}\sum_ {\beta=-M+1}^{0}\frac{\lambda_2^{M(n+1+\frac{{\beta}}{M})}}{M(n+1 +
\frac{\beta}{M} )}
 \frac{\partial^2 F}{\partial{
t_{\beta,n}}\partial{x}}\right).
\end{eqnarray}

The properties of the dBTH mentioned above provide a very sound mathematical background in its possible applications in deriving solutions\cite{kodamaSolution} and some combinatorics meaning in matrix model\cite{Kodama-Pierce}.

\sectionnew{Conclusions and Discussions}
 We define Orlov-Schulman's $\M_L$,
$\M_R$ function of the dBTH and give the  additional symmetry of the
dBTH. We find this kind of Block type Lie algebraic structure is
still preserved by the dBTH. Also we give the quasi-classical limit
relation  of the BTH and  the dBTH and some Hirota bilinear
equations of the dBTH.  We hope the additional symmetry and
dispersionless Hirota bilinear identity of the dBTH can be used more
in other fields of mathematical physics, particularly topological
fields theory and string theory.

The main difference of Block symmetry and HBEs of the dBTH from the ones of the BTH is that the representation space here is a directly a functional space and  the HBEs of the dBTH here are in form of free energy using WKB analysis.

{\bf Acknowledgments} {\noindent \small  This work is supported by NSF of Zhejiang Province under Grant No. LY12A01007,
 NSF of China under Grant No.10971109 and K.C.Wong Magna Fund in
Ningbo University. Jingsong He is also supported by Program for NCET
under Grant No. NCET-08-0515 and NSF of Ningbo under Grant
No.2011A610179. We also thank Professor Yishen Li(USTC,
China) for long-term encouragements and supports. Chuanzhong Li would
like to thank Professor Yuji Kodama in Department of Mathematics at
   Ohio State University for his useful discussions.}

%%%%%%%%%%%%%%%%% References  %%%%%%%%%%%%%%%%%%%%%%%%%%%%%%%%%%%%%%%
\newpage{}

%%%%%%%%%%%%%%%%%%%%%%%%%%%%%%%%%%%%%%%%%%%%%%%%%%%%%%%%%%%%%%%%%%%%%%%%%%%%%%%

%---------------------------------------------------------------------------------------

\end{document}